\documentclass[12pt,notitlepage]{amsart}
\usepackage{amssymb}
\usepackage{amsfonts}
\usepackage{amsmath}
\usepackage{amsxtra}
\usepackage{graphicx}
\usepackage{cite}
\usepackage{footmisc}
\usepackage{enumitem}
\newtheorem{theorem}{Theorem}
\theoremstyle{plain}

\newtheorem{lemma}{Lemma}
\newtheorem{proposition}{Proposition}
\newtheorem{definition}{Definition}
\newtheorem{remark}{Remark}
\newtheorem{example}{Example}
\newtheorem{notation}{Notation}
\newtheorem{assumption}{Assumption}
\numberwithin{equation}{section}
\numberwithin{theorem}{section}
\numberwithin{lemma}{section}
\numberwithin{proposition}{section}
\numberwithin{corollary}{section}
\numberwithin{definition}{section}
\numberwithin{remark}{section}
\numberwithin{example}{section}
\numberwithin{assumption}{section}

\newcommand{\Qp}{\mathbb{Q}_{p}}
\newcommand{\Zp}{\mathbb{Z}_{p}}
\newcommand{\C}{\mathbb{C}}
\newcommand{\R}{\mathbb{R}}
\newcommand{\ii}{\mathrm{i}}
\newcommand{\bJ}{\boldsymbol{J}}
\newcommand{\jj}{\mathfrak{j}}
\newcommand{\bjj}{\boldsymbol{\mathfrak{j}}}
\newcommand{\HH}{\mathfrak{H}}
\newcommand{\F}{\mathcal{F}}
\newcommand{\bH}{\boldsymbol{H}}
\newcommand{\balpha}{\boldsymbol{\alpha}}
\newcommand{\bsigma}{\boldsymbol{\sigma}}
\newcommand{\abs}[1]{\left\vert #1\right\vert}
\newcommand{\norm}[1]{\left\Vert #1\right\Vert}
\newcommand{\bi}{\boldsymbol{i}}
\newcommand{\bj}{\boldsymbol{j}}
\newcommand{\bk}{\boldsymbol{k}}
\newcommand{\bL}{\boldsymbol{L}}
\newcommand{\Om}{\Omega}
\newcommand{\Id}{I_{4}}

\begin{document}

\title[$p$-Adic Dirac Equations, CTQWs, and Quantum Networks]{$p$-Adic Dirac Equations, Continuous-Time Quantum Walks, and Quantum Networks}
\author[Z\'{u}\~{n}iga-Galindo]{W. A. Z\'{u}\~{n}iga-Galindo}
\address{University of Texas Rio Grande Valley\\
School of Mathematical \& Statistical Sciences\\
One West University Blvd\\
Brownsville, TX 78520, United States}
\email{wilson.zunigagalindo@utrgv.edu}
\thanks{The author was partially supported by the Debnath Endowed Professorship}

\begin{abstract}
We introduce a new class of $p$-adic Dirac equations, formulated in the
standard axiomatic framework of quantum mechanics, in which the ordinary
momentum operators are replaced by bounded non-local operators built from
arbitrary integrable radial kernels. We diagonalize the resulting free Dirac
Hamiltonian in momentum space, construct its plane-wave solutions, and
determine its spectrum, which turns out to be pure point, with
eigenvalues of infinite multiplicity; in the massless case we describe the
zero-energy subspace explicitly. We identify the charge-conjugation symmetry
of the model: it is implemented by an antiunitary operator whose square is
$-1$, it exchanges the positive- and negative-energy subspaces of the free
Hamiltonian, and in an external field it maps the charge $e$ and vector
potential $\boldsymbol{A}$ to $-e$ and $-\boldsymbol{A}$. We then discretize
the free equation in two ways, both giving genuine continuous-time quantum
walks rather than discrete-time, coined walks: a first construction on a
countable graph, and a second, more explicit construction on a finite,
tree-like graph. In both cases the discretization is exact: the
walk is the restriction of the $p$-adic Dirac evolution to an invariant
subspace, and its Hamiltonian is a Dirac-type combination of weighted graph
Laplacians with ultrametric weights. We prove that the Born-rule transition
probabilities, summed over the internal degrees of freedom and averaged over
the initial internal state, form a doubly stochastic matrix at every instant;
we also show that, as for any non-trivial continuous-time quantum walk, this
family of matrices is not a Markov semigroup. Building on these results, we
discuss how the construction can serve as the foundation of a quantum network
with four-component spinor degrees of freedom, complementing
earlier, non-relativistic $p$-adic quantum neural networks, and we prove
global well-posedness of the associated nonlinear network equation under
Lipschitz assumptions. We close with a discussion of the open mathematical
and computational problems raised by this construction.
\end{abstract}
\maketitle

%=====================================================================
\section{Introduction}\label{Section_Introduction}

Non-Archimedean, and specifically $p$-adic, models of physical space have been
studied since the late 1980s as a way of exploring what quantum mechanics
might look like if the Planck length is treated not as a mere
order-of-magnitude estimate but as an actual discreteness scale built into the
geometry of space itself,
\cite{Beltrameti-et-al,V-V-QM1,V-V-QM2,V-V-QM3,Volovich,V-V-Z}. Replacing
$\R$ with $\Qp$, the field of $p$-adic numbers, has two
immediate consequences: the resulting space is a totally disconnected,
ultrametric, Cantor-like set rather than a continuum, and the Planck length
can be identified with $p^{-1}$ in a mathematically precise sense. Working
within the Dirac--von Neumann axiomatic formulation of quantum mechanics,
\cite{Dirac,von-Neumann}, on the Hilbert space $L^{2}(\Qp^{N})$,
one obtains a rigorous, non-relativistic $p$-adic quantum mechanics,
\cite{Zuniga-double-slit,Zuniga-Dirac-Causality,Zuniga-2adic,Zuniga-Mayes,Zuniga-Chacon,Zuniga-JR,Zuniga-AP,Zuniga-Wigner}.

Recently, the author and his collaborators have shown that a large class of
$p$-adic Schr\"{o}dinger equations are scaling limits of continuous-time
quantum walks (CTQWs), including CTQWs on graphs,
\cite{Zuniga-2adic,Zuniga-Mayes,Zuniga-Chacon}. This CTQW--Schr\"{o}dinger
correspondence was, in turn, used to build $p$-adic quantum neural networks
(QNNs), whose neurons are organized in the hierarchical, tree-like structure
inherited from $\Zp$ (the unit ball in $\Qp$) rather
than on a Euclidean lattice, \cite{Zuniga-QNN}. Furthermore, the author
recently proposed a new collapse mechanism for wavefunctions (using $(
\R\times\Qp)^{3}$ as model of the physical space, so that the model of the
space-time is $\R\times(\R\times\Qp)^{3}$) that provides a new explanation of
the Wigner's friend paradox, \cite{Zuniga-AP,Zuniga-Wigner}.

The construction of a $p$-adic version of the Dirac equation is very delicate,
due to the relativistic nature of this equation and the fact that using a
space-time (or configuration space) of type $\R\times\Qp^{N}$ is not compatible with special relativity. In
\cite{Zuniga-Dirac-Causality}, the author introduced a $p$-adic Dirac
equation, built from the Vladimirov pseudo-differential operator, which plays
the role of the derivative on $\Qp$, and showed that an isolated $p$-adic
Dirac system does \emph{not} satisfy Einstein causality: the discreteness of
$p$-adic space allows correlations, and hence signals in a suitable sense, to
propagate faster than the (finite, but $p$-dependent) analogue of the speed of
light. The reader may consult \cite{Zuniga-Dirac-Causality,Zuniga-2adic} for a
further discussion on quantum non-locality and faster-than-light
communication.

The $p$-adic Dirac equations are not relativistic: the space-time
$\R\times\Qp^{3}$ carries no analogue of the Lorentz group, and the equations
are not invariant under any such group. They do, however, retain the
algebraic and spectral structure of the Dirac equation. They are first-order
equations for four-component spinors, whose free Hamiltonian is built from
the Dirac matrices; they admit plane-wave solutions with the dispersion
relation $E^{2}=\abs{\bjj(k)}^{2}+m^{2}$, where $\bjj(k)\in\R^{3}$ is the
symbol of the operators that replace the momentum (Section
\ref{Section_new_type}); their spectrum splits into positive-
and negative-energy branches, which are exchanged by a charge-conjugation
symmetry; and a Foldy--Wouthuysen transformation decouples them into a pair of
square-root Klein--Gordon equations. It is in this sense that we call them
Dirac equations. A topologically protected extension
of the equation of \cite{Zuniga-Dirac-Causality}, modeled on the
Jackiw--Rebbi soliton, was introduced more recently in \cite{Zuniga-JR}.

Until now, these two threads---the $p$-adic Dirac equation on one side, the
$p$-adic CTQW\slash QNN formalism on the other---have remained disjoint: no
discretization of a $p$-adic Dirac equation into a continuous-time quantum
walk had been constructed, and consequently no $p$-adic quantum network with
spinor-like internal degrees of freedom was available. In the ordinary
(Archimedean) literature, which we review in Section
\ref{Section_Related_Work}, the overwhelming majority of the works relating
quantum walks to the Dirac equation use \emph{discrete-time}, coined quantum
walks, going back to Feynman's relativistic checkerboard and made rigorous by
Bialynicki-Birula, Meyer, and Succi and Benzi, with the Dirac equation
appearing as a continuum \emph{limit} of the walk,
\cite{Bialynicki-Birula,Meyer,Succi-Benzi,Arrighi-Nesme-Forets}. Continuous-time
constructions carrying a spinor-like structure also exist---for instance the
CTQWs with internal degrees of freedom on periodic lattices used by Childs and
Goldstone for spatial search, whose Hamiltonian is a lattice discretization of
the Dirac operator \cite{Childs-Goldstone}, the topological Dirac operator on
networks and simplicial complexes introduced by Bianconi \cite{Bianconi},
CTQWs on bipartite graphs \cite{Todtli}, and tight-binding Hamiltonians on the
honeycomb lattice \cite{Wallace,Castro-Neto}---but none of them is formulated
over a hierarchical, ultrametric graph.

The present paper fills this gap in the non-Archimedean setting. We introduce
a different, complementary class of $p$-adic Dirac equations, in which the
momentum operator along each spatial direction $x_{l}$ (the analogue of
$-\ii\partial/\partial x_{l}$) is replaced by a non-local difference operator
$\bJ_{x_{l}}$ built from an \emph{arbitrary} integrable radial kernel
$J_{l}(\abs{y}_{p})$ (consequently, the free Hamiltonian is
$\balpha\cdot\nabla+\beta m$, without the factor $-\ii$ of the textbook form
$-\ii\balpha\cdot\nabla_{\R}+\beta m$, which is already absorbed into the
operators $\bJ_{x_{l}}$; see Remark \ref{Remark_no_minus_i}), rather than from the Vladimirov pseudo-differential
operator used in \cite{Zuniga-Dirac-Causality}. Unlike the Vladimirov
operator, which is unbounded on $L^{2}(\Qp)$, the operator $\bJ_{x_{l}}$ is
bounded for every integrable kernel $J_{l}$; consequently, the equation of
\cite{Zuniga-Dirac-Causality} is \emph{not} a particular case of the family
introduced here, nor is the present family a generalization of it in the
strict sense---the two constructions are built from operators of genuinely
different analytic type, neither contained in the other. The boundedness of
$\bJ_{x_{l}}$ is, however, precisely what makes this new family well suited to
discretization: it yields a large, explicit class of $p$-adic Dirac equations,
parametrized by the choice of $J_{1}$, $J_{2}$, $J_{3}$, each member of which
can be discretized directly into a continuous-time quantum walk.

We carry out, for this new class, the programme of analysis that is classical
for the standard Dirac equation, \cite{Thaller,Bjorken,Greiner}: we
diagonalize the free Hamiltonian $\bH_{0}$ in the Fourier representation,
construct its plane-wave solutions, determine its spectrum and show that it
is contained in the spectrum of the standard free Dirac operator, exhibit the
Foldy--Wouthuysen-type transformation that decouples $\bH_{0}$ into a pair of
square-root Klein--Gordon operators, and study charge conjugation. Two
features distinguish the $p$-adic model from the standard one. First, the
spectrum of $\bH_{0}$ is pure point, every eigenvalue has infinite
multiplicity, and in the massless case $\bH_{0}$ may have an
infinite-dimensional kernel (Theorem \ref{Theorem_spectrum} and Proposition
\ref{Prop_zero_modes}). Second, since the operators $\bJ_{x_{l}}$ commute with
complex conjugation (whereas $-\ii\partial/\partial x_{l}$ anticommutes with it),
the charge-conjugation operator is $\mathcal{C}\Psi=\alpha_{2}\overline{\Psi}$,
which satisfies $\mathcal{C}^{2}=-1$, and in the presence of an
electromagnetic field it reverses both the charge and the vector potential
(Lemma \ref{Lemma_charge_conj}).

We then construct two discretizations of the free equation into
CTQWs: a first construction on the countable graph
$\widetilde{G}_{n}^{3}=\Qp^{3}/p^{n}\Zp^{3}$ (Section \ref{Section_CTQW_I}),
and a second, more explicit construction on the finite graph
$G_{n}^{3}=\Zp^{3}/p^{n}\Zp^{3}$ (Section \ref{Section_CTQW_II}). The
key observation is that the space of square-integrable functions that are
constant on the balls of radius $p^{-n}$ is invariant under $\bH_{0}$, and
that, on this space, each $\bJ_{x_{l}}$ acts as a weighted graph Laplacian
whose weights are determined by the $p$-adic tree. Consequently, the
restriction of the $p$-adic Dirac evolution to this space is exactly a CTQW, in
the sense of Farhi and Gutmann \cite{Farhi-Gutman}, with a $\C^{4}$ internal
degree of freedom at each vertex; no continuum limit is involved. For these
walks we prove that the Born-rule transition probabilities, summed over the
four spinor components, define a stochastic matrix for every $t$, and that,
after averaging over the initial internal state, this matrix is doubly
stochastic (Theorem \ref{Theorem_stochastic}). As for any CTQW, the
resulting family of stochastic matrices is not a Markov semigroup
(Proposition \ref{Prop_non_Markov}).

Beyond their intrinsic mathematical interest, these results have a direct
implication for quantum networks, which we develop in Section
\ref{Section_Quantum_Computers}. The constructions of Sections
\ref{Section_CTQW_I} and \ref{Section_CTQW_II} suggest that the free $p$-adic
Dirac CTQW can serve as the substrate of a quantum network in the sense of
\cite{Zuniga-QNN}, extended here from scalar, non-relativistic internal
degrees of freedom to four-component spinor degrees of freedom. We write
the corresponding nonlinear network equation precisely and prove its global
well-posedness under Lipschitz assumptions (Proposition
\ref{Prop_well_posed}). We close the paper, in Section
\ref{Section_Discussion}, with a discussion of the mathematical open problems
that this construction raises.

In a forthcoming paper, we plan to extend the numerical methods introduced in
\cite{Zuniga-Chacon} and \cite{Zuniga-QNN} to the case of the $p$-adic Dirac
equations and the associated QNNs.

The paper is organized as follows. Section \ref{Section_Related_Work} reviews
the state of the art on the relation between quantum walks and the Dirac
equation, both in the Archimedean and in the non-Archimedean setting. Section
\ref{Section_Some_Preliminaries} collects the notation and basic facts on
$p$-adic analysis used throughout the paper (a more detailed review is given
in the Appendix). Section \ref{Section_New_Class} introduces the new class of
$p$-adic Dirac equations. Section \ref{Section_Plane_Waves} constructs their
plane-wave solutions. Section \ref{Section_3} diagonalizes the free
Hamiltonian, determines its spectrum, describes the positive-, negative- and
zero-energy subspaces, and establishes the Foldy--Wouthuysen decoupling.
Section \ref{Section_4} studies charge conjugation. Sections
\ref{Section_CTQW_I} and \ref{Section_CTQW_II} construct the two CTQW
discretizations of the free equation. Section \ref{Section_Quantum_Computers}
discusses the resulting quantum-network construction, and Section
\ref{Section_Discussion} concludes with a discussion of open problems. An
Appendix collects the background material on $p$-adic analysis needed in the
body of the paper.

%=====================================================================
\section{Related work}\label{Section_Related_Work}

The idea that the Dirac equation is, at bottom, a quantum walk on a lattice
goes back to Feynman's relativistic checkerboard model for the $1+1$
dimensional propagator, \cite{Feynman-Hibbs}. This intuition was made
mathematically precise, independently, by Bialynicki-Birula,
\cite{Bialynicki-Birula}, Meyer, \cite{Meyer}, and Succi and Benzi,
\cite{Succi-Benzi}, who showed that a homogeneous (translation-invariant,
time-independent), strictly causal discrete-time quantum walk (DTQW) on
$\mathbb{Z}$, with a two-dimensional coin space attached to each site,
converges, as the lattice spacing $\varepsilon\rightarrow0$, to the exact
solution of the free Dirac equation. Arrighi, Nesme, and Forets extended this
construction to $2+1$ and $3+1$ dimensions and proved consistency, stability,
and convergence of the corresponding DTQWs to the higher-dimensional Dirac
equation, \cite{Arrighi-Nesme-Forets}. This body of work established what is
now the standard paradigm: \emph{the Dirac equation is the continuum limit of
a discrete-time, coined quantum walk}, with the coin space playing the role of
the spinor degrees of freedom.

A second, closely related line of research uses this correspondence to locate
genuinely relativistic phenomena---Zitterbewegung and the Klein
paradox---inside the purely quantum-informational dynamics of the DTQW.
Strauch obtained closed-form Bessel-function solutions of the one-dimensional
DTQW and showed that, in one limit, they reduce to the wave-packet solutions
of the free Dirac equation, exhibiting relativistic spreading, while in the
opposite limit they reduce to the solutions of the ordinary continuous-time
quantum walk (CTQW), \cite{Strauch1,Strauch2}. Kurzy\'{n}ski compared the
one-dimensional DTQW directly with the discretized Dirac equation and showed
that both Zitterbewegung and Klein's paradox are already present in the walk,
\cite{Kurzynski}. Chandrashekar, Banerjee, and Srikanth pushed this comparison
further, relating the coupled and decoupled forms of the DTQW to the free
Dirac and Klein--Gordon equations, respectively, and identifying the walk's
coin as the analogue of the spinor degree of freedom, \cite{Chandrashekar}.
These results were subsequently realized experimentally on a trapped-ion
platform: Gerritsma, Kirchmair, Z\"{a}hringer, Solano, Blatt, and Roos
implemented the $1+1$ dimensional Dirac equation by coupling a single qubit to
the vibrational mode of a trapped ion and observed Zitterbewegung directly,
\cite{Gerritsma1}, and the same group later used the same platform to simulate
the Klein paradox, \cite{Gerritsma2}.

A parallel line of research studies \emph{continuous}-time quantum walks. The
CTQW was introduced by Farhi and Gutmann as the direct quantization of a
classical continuous-time random walk on a graph, \cite{Farhi-Gutman}: its
dynamics is governed by a Schr\"{o}dinger equation $\ii\partial_{t}\psi=H\psi$,
with $H$ the graph adjacency matrix or the graph Laplacian. The framework was
subsequently developed and surveyed extensively, e.g., in
\cite{Mulkne-Blumen,Venegas-Andraca,Childs-et-al}. In \cite{Zuniga-2adic}, the
author established that the Farhi--Gutmann CTQWs are a particular case of a
much larger family of CTQWs constructed using techniques of $p$-adic analysis,
as we discuss below. For this reason the CTQW is naturally the continuous-time
counterpart of the \emph{Schr\"{o}dinger}, rather than the Dirac, equation,
and the literature connecting CTQWs directly to Dirac dynamics is thinner than
that for its discrete-time counterpart.

Several constructions do, however, inject a genuine spinor-like structure into
continuous-time dynamics on graphs. Childs and Goldstone \cite{Childs-Goldstone}
showed that the continuous-time quantum-walk search algorithm on low-dimensional
periodic lattices can be substantially improved by enlarging the Hilbert space
with internal (spin) degrees of freedom; their walk Hamiltonian is a lattice
discretization of the massless Dirac operator, and the Dirac equation governs
its long-wavelength behaviour. Bianconi \cite{Bianconi} introduced the
topological Dirac operator of a network (and, more generally, of a simplicial
complex), which acts on ``topological spinors'' formed by node and link
signals, and the associated topological Dirac equation, which generates a
continuous-time unitary dynamics on an arbitrary graph. T\"{o}dtli, Laner,
Semenov, Paoli, Blattner, and Kunegis studied CTQWs on directed bipartite
graphs and showed that probability transport between the two node
partitions---which can be read as the two components of a spinor---can be
tuned continuously, and even completely suppressed, by a single phase
parameter, \cite{Todtli}. Finally, continuous-time tight-binding Hamiltonians
on the honeycomb lattice, introduced to model electron transport in graphite
and graphene \cite{Wallace,Castro-Neto}, provide another continuous-time
realization of Dirac dynamics, since their low-energy dispersion around the
Dirac points is that of the $2+1$ dimensional massless Dirac equation; the
corresponding \emph{discrete-time} quantum walks on triangular and honeycomb
lattices, and their continuum limits, were constructed by Jay, Debbasch, and
Wang \cite{Jay-Debbasch-Wang}.

Di Molfetta and Debbasch classified which continuum limits are attainable from
families of coined walks and showed that, when a continuum limit exists, it is
described either by a Dirac-like equation or by a pair of decoupled
Klein--Gordon equations, \cite{DiMolfetta-Debbasch}; the same authors,
together with Brachet, extended this program to curved space-time, showing
that a suitably chosen family of DTQWs converges to massless Dirac fermions
propagating on a given background metric, \cite{DiMolfetta-Brachet-Debbasch}.
Arnault and Debbasch developed gauge-covariant discrete-time walks
\cite{Arnault-Debbasch}.

None of the constructions reviewed above is formulated over a non-Archimedean
field, or over a hierarchical (ultrametric) graph. On the $p$-adic side,
the author showed that free $p$-adic Schr\"{o}dinger equations arise as the
scaling limit of continuous-time quantum Markov chains on the $p$-adic tree
$\Zp$, and that discretizing these equations produces genuine CTQWs on finite
graphs, \cite{Zuniga-2adic}; this correspondence was developed further,
together with Chac\'{o}n-Cort\'{e}s, through detailed comparisons and
numerical simulations of the resulting discretizations, \cite{Zuniga-Chacon},
and, together with Mayes, in the setting of $p$-adic infinite potential wells,
\cite{Zuniga-Mayes}. On the equation side, a $p$-adic Dirac equation was
introduced in \cite{Zuniga-Dirac-Causality}, where it was shown to violate
Einstein causality as a consequence of the discreteness of $p$-adic space; a
topologically protected version of this equation, modeled on the
Jackiw--Rebbi soliton, was introduced more recently in \cite{Zuniga-JR}. These
two threads---$p$-adic CTQWs on one side, the $p$-adic Dirac equation on the
other---have, until now, remained disjoint: no discretization of a $p$-adic
Dirac equation into a continuous-time quantum walk has been constructed. The
present paper fills this gap. In contrast with the Archimedean
constructions above, where the Dirac equation appears as a long-wavelength or
continuum limit of the walk, here the walk is the \emph{exact} restriction of
the $p$-adic Dirac evolution to an invariant subspace (Theorems
\ref{Theorem_CTQW_I} and \ref{Theorem_CTQW_II}). We note that a companion
construction, restricted to the $p$-adic Schr\"{o}dinger case and aimed at
quantum neural network architectures, was developed in \cite{Zuniga-QNN}.

%=====================================================================
\section{Some mathematical preliminaries}\label{Section_Some_Preliminaries}

From now on, we use $p$ to denote a fixed prime number. Any non-zero $p$-adic
number $x$ has a unique expansion of the form
\begin{equation}
x=x_{-k}p^{-k}+x_{-k+1}p^{-k+1}+\ldots+x_{0}+x_{1}p+\ldots,
\label{p-adic-number}
\end{equation}
with $x_{-k}\neq0$, where $k$ is an integer, and the $x_{j}$s are numbers
from the set $\{0,1,\ldots,p-1\}$. The set of all possible
sequences of the form (\ref{p-adic-number}) constitutes the field of $p$-adic
numbers $\Qp$. There are natural field operations, sum and
multiplication, on series of form (\ref{p-adic-number}). There is also a norm
in $\Qp$ defined as $\abs{x}_{p}=p^{-ord(x)}$,
where $ord(x)=-k$, for a nonzero $p$-adic number $x$. By definition
$ord(0)=\infty$. The field of $p$-adic numbers with the distance induced by
$\abs{\cdot}_{p}$ is a complete ultrametric space. The
ultrametric property refers to the fact that $\abs{x-y}_{p}\leq\max\{\abs{x-z}_{p},\abs{z-y}_{p}\}$
for any $x$, $y$, $z$ in $\Qp$. The $p$-adic
integers are sequences of form (\ref{p-adic-number}) with $-k\geq0$. All these
sequences constitute the unit ball $\Zp$. The unit ball is an
infinite rooted tree with fractal structure. As a topological space
$\Qp$ is homeomorphic to a Cantor-like subset of the real line,
see, e.g., \cite{V-V-Z,Alberioetal}.

We extend the $p$-adic norm to $\Qp^{N}$ by taking
\[
\norm{x}_{p}:=\max_{1\leq i\leq N}\abs{x_{i}}_{p},\qquad\text{for }x=(x_{1},\dots,x_{N})\in\Qp^{N}.
\]
We define $ord(x)=\min_{1\leq i\leq N}\{ord(x_{i})\}$, then $\norm{x}_{p}=p^{-ord(x)}$. The metric space $(\Qp^{N},\norm{\cdot}_{p})$ is a complete ultrametric space.

A function $\varphi:\Qp^{N}\rightarrow\C$ is called locally
constant, if for any $a\in\Qp^{N}$, there is an integer $l=l(a)$,
such that
$\varphi(a+x)=\varphi(a)$ for any $\norm{x}_{p}\leq p^{l}$.
If $\varphi$ is locally constant and has compact support, we say that
$\varphi$ is a test function. We denote by $\mathcal{D}(\Qp^{N})$ the complex
vector space of test functions. There is a natural integration theory so that
$\int_{\Qp^{N}}\varphi(x)d^{N}x$ gives a well-defined complex number. The
measure $d^{N}x$ is the Haar measure of $\Qp^{N}$, normalized by
$\int_{\Zp^{N}}d^{N}x=1$. In the Appendix, we give a quick review of the basic
aspects of the $p$-adic analysis required here.

By $p$-adic quantum mechanics (QM), we mean QM in the sense of the Dirac--von
Neumann formulation on the Hilbert space
\[
L^{2}(\Qp^{N}):=L^{2}(\Qp^{N},d^{N}x)=\left\{
f:\Qp^{N}\rightarrow\C;\norm{f}_{2}=\left(
\int_{\Qp^{N}}\abs{f(x)}^{2}d^{N}x\right)^{\frac{1}{2}}<\infty\right\}.
\]
Given $f,g\in L^{2}(\Qp^{N})$, we set
$\langle f,g\rangle=\int_{\Qp^{N}}f(x)\overline{g(x)}d^{N}x$,
where the bar denotes the complex conjugate.

%=====================================================================
\section{A new class of $p$-adic Dirac equations}\label{Section_New_Class}

\subsection{Abstract Dirac equations}

We denote the Pauli matrices as
\[
\sigma_{1}=\begin{bmatrix}0&1\\1&0\end{bmatrix},\quad
\sigma_{2}=\begin{bmatrix}0&-\ii\\ \ii&0\end{bmatrix},\quad
\sigma_{3}=\begin{bmatrix}1&0\\0&-1\end{bmatrix},
\]
where $\ii=\sqrt{-1}\in\C$, and the $4\times4$ Dirac matrices as
\[
\beta=\begin{bmatrix}\boldsymbol{1}&\boldsymbol{0}\\\boldsymbol{0}&-\boldsymbol{1}\end{bmatrix},\qquad
\alpha_{l}=\begin{bmatrix}\boldsymbol{0}&\sigma_{l}\\\sigma_{l}&\boldsymbol{0}\end{bmatrix},\quad l=1,2,3,
\]
where $\boldsymbol{1}$ denotes the $2\times2$ identity matrix, and
$\boldsymbol{0}$ denotes the $2\times2$ zero matrix. These matrices are
Hermitian and satisfy the anticommutation relations
\begin{equation}
\alpha_{l}\alpha_{k}+\alpha_{k}\alpha_{l}=2\delta_{lk}\Id,\qquad
\alpha_{l}\beta+\beta\alpha_{l}=0,\qquad\beta^{2}=\Id,
\label{Clifford}
\end{equation}
where $\Id$ is the $4\times4$ identity matrix; note that $\Id=\operatorname{diag}(\boldsymbol{1},\boldsymbol{1})$ (block-diagonal), where $\boldsymbol{1}$ is the $2\times2$ identity matrix (identity operators on other spaces are denoted by $\mathrm{Id}$). Note that $\alpha_{1}$,
$\alpha_{3}$, $\beta$ are real, while $\overline{\alpha_{2}}=-\alpha_{2}$.

We set
\[
\nabla:=\begin{bmatrix}\bJ_{x_{1}}\\\bJ_{x_{2}}\\\bJ_{x_{3}}\end{bmatrix},
\]
where $\bJ_{x_{1}}$, $\bJ_{x_{2}}$, $\bJ_{x_{3}}$ are pairwise commuting
self-adjoint operators on $L^{2}(\mathbb{X}^{3})$, which play the role of an
abstract version of the momentum operators $-\ii\frac{\partial}{\partial x_{l}}$.
The spatial variables $x_{1}$, $x_{2}$, $x_{3}$ run in an abstract space
$\mathbb{X}$, which can be specialized to $\R$. We use the standard notation
\[
\balpha\cdot\nabla:=\sum_{l=1}^{3}\alpha_{l}\bJ_{x_{l}},\qquad
\bsigma\cdot\nabla:=\sum_{l=1}^{3}\sigma_{l}\bJ_{x_{l}}.
\]
An abstract Dirac equation is an equation of the form
\begin{equation}
\ii\frac{\partial}{\partial t}\Psi(t,x)=\left(\balpha\cdot\nabla+\beta m\right)\Psi(t,x),\qquad
t\in\R,\ x\in\mathbb{X}^{3}, \label{Dirac_1}
\end{equation}
where $m\geq0$ (the mass), and
$\Psi^{T}(t,x)=\left[\Psi_{1}(t,x)\ \Psi_{2}(t,x)\ \Psi_{3}(t,x)\ \Psi_{4}(t,x)\right]\in\C^{4}$.
If $\mathbb{X}=\R$ and $\bJ_{x_{l}}=-\ii\frac{\partial}{\partial x_{l}}$,
then (\ref{Dirac_1}) is the standard free Dirac equation (in units with
$\hbar=c=1$).

We define the free Hamiltonian as the operator
\[
\bH_{0}=\balpha\cdot\nabla+\beta m=\begin{bmatrix}m\boldsymbol{1}&\bsigma\cdot\nabla\\\bsigma\cdot\nabla&-m\boldsymbol{1}\end{bmatrix}.
\]
Then, our abstract version of the Dirac equation can be rewritten as
\begin{equation}
\ii\frac{\partial}{\partial t}\Psi(t,x)=\bH_{0}\Psi(t,x),\qquad t\in\R,\ x\in\mathbb{X}^{3}. \label{Equ_dirac}
\end{equation}

\begin{remark}[On the absence of the factor $-\ii$ in $\bH_{0}$]\label{Remark_no_minus_i}
The reader used to the standard Dirac equation should note that our free
Hamiltonian is written as $\bH_{0}=\balpha\cdot\nabla+\beta m$, \emph{without}
the factor $-\ii$ that appears in the textbook form
$\bH_{0}^{Arch}=-\ii\balpha\cdot\nabla_{\R}+\beta m$, where
$\nabla_{\R}=(\partial_{x_{1}},\partial_{x_{2}},\partial_{x_{3}})$; see, e.g.,
\cite[Chapter 1]{Thaller}. This is not an omission: in this article the
components of $\nabla$ are the analogues of the self-adjoint momentum operators
$-\ii\partial/\partial x_{l}$, not of the skew-adjoint operators
$\partial/\partial x_{l}$, and the factor $-\ii$ is already absorbed into them.
Thus, for $\mathbb{X}=\R$, one recovers $\bH_{0}^{Arch}$ by taking
$\bJ_{x_{l}}=-\ii\partial/\partial x_{l}$.

This convention is forced in the $p$-adic setting. The operators
$\bJ_{x_{l}}$ introduced in Section \ref{Section_new_type} have real, even
kernels; they are self-adjoint (Section \ref{Section_J_operator}), like $-\ii\partial/\partial x_{l}$.
If one inserted an additional factor $-\ii$, i.e., used
$-\ii\balpha\cdot\nabla+\beta m$, the resulting operator would not be
self-adjoint, since $(-\ii\alpha_{l}\bJ_{x_{l}})^{\ast}=\ii\alpha_{l}\bJ_{x_{l}}$.
Indeed, its symbol $\widetilde{h}(\xi)=-\ii\balpha\cdot\bjj(\xi)+\beta m$
(notation of Section \ref{Section_new_type}) satisfies, by (\ref{Clifford}),
$\widetilde{h}(\xi)^{2}=\left(m^{2}-\abs{\bjj(\xi)}^{2}\right)\Id$, so its
eigenvalues $\pm\sqrt{m^{2}-\abs{\bjj(\xi)}^{2}}$ are purely imaginary whenever
$\abs{\bjj(\xi)}>m$ (for instance, at every $\xi$ with $\bjj(\xi)\neq0$ when
$m=0$), and the corresponding evolution $e^{-\ii t\widetilde{h}(\xi)}$ grows
exponentially in $t$. The associated dynamics would not be unitary, and the
probabilistic interpretation of Sections \ref{Section_CTQW_I} and
\ref{Section_CTQW_II} would be lost.

Equivalently, if one prefers to keep the textbook form of the Hamiltonian, one
may introduce the $p$-adic analogues of the derivatives,
$\mathcal{D}_{x_{l}}:=\ii\bJ_{x_{l}}$, which are skew-adjoint, like
$\partial/\partial x_{l}$; then
$-\ii\sum_{l}\alpha_{l}\mathcal{D}_{x_{l}}+\beta m=\balpha\cdot\nabla+\beta m=\bH_{0}$.
The two descriptions define the same operator; we use the one in terms of the
self-adjoint operators $\bJ_{x_{l}}$ throughout. Finally, we note that, since
the $\bJ_{x_{l}}$ commute with complex conjugation (Remark
\ref{Remark_nature_J}) while $-\ii\partial/\partial x_{l}$ anticommutes with it,
the charge-conjugation operator of the $p$-adic model differs from the standard
one; see Section \ref{Section_4}.
\end{remark}

\subsection{The $\bJ$ operator}\label{Section_J_operator}

We take $\mathbb{X}=\Qp$, the field of $p$-adic numbers. In the
Appendix, we give a brief review of the essential ideas and results of
$p$-adic analysis required in this paper.

We fix a function $J:[0,\infty)\rightarrow\R$ such that $y\mapsto
J(\abs{y}_{p})$ is an integrable function on $\Qp$, and set
\[
\norm{J}_{1}:=\int_{\Qp}\abs{J(\abs{y}_{p})}dy<\infty,\qquad
\widehat{J}(0):=\int_{\Qp}J(\abs{y}_{p})dy\in\R.
\]
Here $dy$ denotes the normalized Haar measure on $\Qp$; see the Appendix.
Notice that $\norm{J}_{1}$ is the $L^{1}$-norm of $J(\abs{y}_{p})$, while
$\widehat{J}(0)$ is the signed integral of $J(\abs{y}_{p})$; no sign condition
is imposed on $J$, so $\widehat{J}(0)$ may be negative or zero, but always
$\abs{\widehat{J}(0)}\leq\norm{J}_{1}$.

Now, for $\rho\in[1,\infty]$, we define the operator
\[
\begin{array}{llll}
\bJ: & L^{\rho}(\Qp) & \rightarrow & L^{\rho}(\Qp)\\[2pt]
& \phi(x) & \mapsto & \displaystyle\int_{\Qp}J(\abs{x-y}_{p})\left\{\phi(y)-\phi(x)\right\}dy,
\end{array}
\]
where the integral is understood as
$\bJ\phi(x)=(J(\abs{\cdot}_{p})\ast\phi)(x)-\widehat{J}(0)\phi(x)$; both
terms are well defined for almost every $x$. By Young's inequality,
$\norm{J(\abs{\cdot}_{p})\ast\phi}_{\rho}\leq\norm{J}_{1}\norm{\phi}_{\rho}$,
and hence
\[
\norm{\bJ\phi}_{\rho}\leq\norm{J}_{1}\norm{\phi}_{\rho}+\abs{\widehat{J}(0)}\norm{\phi}_{\rho}\leq2\norm{J}_{1}\norm{\phi}_{\rho},
\]
for $\rho\in[1,\infty]$. Therefore $\bJ$ is a linear bounded operator
from $L^{\rho}(\Qp)$ into itself for $\rho\in[1,\infty]$. For instance, the kernel $J(\abs{y}_{p})=-\Om(\abs{y}_{p})$, the
negative of the characteristic function of $\Zp$, is admissible; in this case
$\widehat{J}(0)=-1$, $\norm{J}_{1}=1$, and the above bound reads
$\norm{\bJ\phi}_{\rho}\leq2\norm{\phi}_{\rho}$. Most of the time, we will take
$\rho=2$.

In the case $\rho=2$, the operator $\bJ$ is pseudo-differential. Indeed,
the Fourier transform of $J(\abs{x}_{p})$,
\[
\F_{x\rightarrow\xi}(J(\abs{x}_{p}))=\int_{\Qp}J(\abs{x}_{p})\chi_{p}(\xi x)dx=\widehat{J}(\abs{\xi}_{p}),
\]
is a continuous, real-valued, radial function, which tends to zero as
$\abs{\xi}_{p}\rightarrow\infty$ (Riemann--Lebesgue lemma), and whose value at
$\xi=0$ is $\widehat{J}(0)$. It is real-valued because $J(\abs{-x}_{p})=J(\abs{x}_{p})$,
and radial because the Haar measure and $\abs{\cdot}_{p}$ are invariant under
$x\mapsto ax$ for $\abs{a}_{p}=1$. We set
\begin{equation}
\jj(\xi):=\widehat{J}(\abs{\xi}_{p})-\widehat{J}(0),\qquad\xi\in\Qp.
\label{Symbol_j}
\end{equation}
The function $\jj$ is real-valued, radial, continuous, $\jj(0)=0$, and
$\abs{\jj(\xi)}\leq2\norm{J}_{1}$. Furthermore, for $\phi\in L^{2}(\Qp)$,
\[
\bJ\phi(x)=\F_{\xi\rightarrow x}^{-1}\left\{\jj(\xi)\widehat{\phi}(\xi)\right\}\in L^{2}(\Qp),
\]
where $\widehat{\phi}(\xi)$ denotes the Fourier transform of $\phi(x)$; in
particular, the operator norm of $\bJ$ on $L^{2}(\Qp)$ is
$\sup_{\xi}\abs{\jj(\xi)}\leq2\norm{J}_{1}$. A particularly important formula
is
\begin{equation}
\bJ\chi_{p}(kx)=\jj(k)\chi_{p}(kx)=\left\{\widehat{J}(\abs{k}_{p})-\widehat{J}(0)\right\}\chi_{p}(kx),
\qquad\text{for } k,\ x\in\Qp. \label{Formula 1}
\end{equation}
Here, $\chi_{p}$ denotes the standard additive character of $(\Qp,+)$, which
is the $p$-adic analogue of the complex exponential function; see the
Appendix. The function $x\mapsto\chi_{p}(kx)$ does not belong to
$L^{2}(\Qp)$, but it is bounded, so the integral defining $\bJ\chi_{p}(kx)$
converges absolutely; indeed, changing variables $y=x-z$,
\[
\bJ\chi_{p}(kx)=\chi_{p}(kx)\int_{\Qp}J(\abs{z}_{p})\left\{\chi_{p}(-kz)-1\right\}dz=\chi_{p}(kx)\left\{\widehat{J}(\abs{k}_{p})-\widehat{J}(0)\right\}.
\]

The operator $\bJ$ is self-adjoint in $L^{2}(\Qp)$: since the
Fourier transform preserves the inner product and $\jj$ is real-valued,
\begin{align*}
\langle\psi,\bJ\phi\rangle&=\int_{\Qp}\psi(x)\overline{\bJ\phi(x)}dx=\int_{\Qp}\widehat{\psi}(\xi)\overline{\F_{x\rightarrow\xi}\{\bJ\phi(x)\}}d\xi
=\int_{\Qp}\widehat{\psi}(\xi)\,\jj(\xi)\,\overline{\widehat{\phi}(\xi)}d\xi\\
&=\int_{\Qp}\bJ\psi(x)\overline{\phi(x)}dx=\langle\bJ\psi,\phi\rangle.
\end{align*}
Along the paper, we use the symbol $\langle\cdot,\cdot\rangle$ to
denote the inner product of $L^{2}(\Qp^{N})$. The
corresponding norm is denoted as $\norm{\cdot}_{2}$.

\begin{example}\label{Example_kernel}
Take $J(\abs{y}_{p})=\Om(\abs{y}_{p})$, the characteristic function of $\Zp$.
Then $\widehat{J}(\abs{\xi}_{p})=\Om(\abs{\xi}_{p})$, $\widehat{J}(0)=1$,
$\jj(\xi)=-1_{\Qp\smallsetminus\Zp}(\xi)$, and
$\bJ\phi(x)=\int_{x+\Zp}\phi(y)dy-\phi(x)$, i.e., $\bJ$ is minus the
orthogonal projection onto the functions whose Fourier transform is supported
outside $\Zp$. If $J$ is supported in the ball $B_{L}=p^{-L}\Zp$, then
$\chi_{p}(\xi y)=1$ for $\abs{\xi}_{p}\leq p^{-L}$ and $y\in B_{L}$, hence
$\widehat{J}(\abs{\xi}_{p})=\widehat{J}(0)$ and $\jj(\xi)=0$ for
$\abs{\xi}_{p}\leq p^{-L}$.
\end{example}

\begin{remark}\label{Remark_nature_J}
On $\R$, the momentum operator $-\ii\frac{d}{dx}$ is a Fourier multiplier with
a real, \emph{odd} symbol. The operator $\bJ$ is a Fourier multiplier with a
real, \emph{even} (indeed radial) symbol $\jj$. Since $\Qp$ is not an ordered
field, radial symbols are the natural ones in the $p$-adic setting: for
instance, the Vladimirov operator, used in \cite{Zuniga-Dirac-Causality}, has
symbol $\abs{\xi}_{p}$, which is the analogue of the symbol of $\abs{d/dx}$.
In the abstract equation (\ref{Dirac_1}), the operators $\bJ_{x_{l}}$ enter
only through the identity $(\balpha\cdot\nabla)^{2}=\sum_{l}\bJ_{x_{l}}^{2}$,
which follows from (\ref{Clifford}) and the commutativity of the
$\bJ_{x_{l}}$; this is what produces the Dirac-type dispersion relation of
Section \ref{Section_Plane_Waves}. The parity of the symbol, on the other
hand, determines the behaviour under complex conjugation: since the kernel
$J$ is real, $\overline{\bJ\phi}=\bJ\overline{\phi}$, whereas
$\overline{-\ii\phi^{\prime}}=-(-\ii\overline{\phi}^{\prime})$. This difference
has consequences for charge conjugation; see Section \ref{Section_4}.
\end{remark}

\subsection{A new type of $p$-adic Dirac equations}\label{Section_new_type}

We now fix three functions $J_{l}:[0,\infty)\rightarrow\R$, $l=1,2,3$, such
that $J_{l}(\abs{y}_{p})$ is an integrable function on $\Qp$, and set

\begin{equation}
\bJ_{x_{l}}\phi(x):=\int_{\Qp}J_{l}(\abs{x_{l}-y}_{p})\left\{\phi(x^{(l)}(y))-\phi(x)\right\}dy,\qquad l=1,2,3,
\label{J_x_l}
\end{equation}
for $\phi\in L^{2}(\Qp^{3})$, where $x=(x_{1},x_{2},x_{3})$ and $x^{(l)}(y)$
denotes the point obtained from $x$ by replacing its $l$-th coordinate by $y$.
In other words, $\bJ_{x_{l}}$ is the operator $\bJ$ of the previous subsection,
built from the kernel $J_{l}$, acting on the variable $x_{l}$. We denote by
$\jj_{l}(\xi_{l})=\widehat{J}_{l}(\abs{\xi_{l}}_{p})-\widehat{J}_{l}(0)$ the
corresponding symbols, and set $\nabla^{T}:=\left[\bJ_{x_{1}}\ \bJ_{x_{2}}\ \bJ_{x_{3}}\right]$.

\noindent\textbf{Notation.} For $\xi=(\xi_{1},\xi_{2},\xi_{3})\in\Qp^{3}$, we set
\[
\bjj(\xi):=\left(\jj_{1}(\xi_{1}),\jj_{2}(\xi_{2}),\jj_{3}(\xi_{3})\right)\in\R^{3},\qquad
\abs{\bjj(\xi)}^{2}:=\sum_{l=1}^{3}\jj_{l}(\xi_{l})^{2},
\]
\[
\balpha\cdot\bjj(\xi):=\sum_{l=1}^{3}\alpha_{l}\,\jj_{l}(\xi_{l}),\qquad
\bsigma\cdot\bjj(\xi):=\sum_{l=1}^{3}\sigma_{l}\,\jj_{l}(\xi_{l}),
\]
and
\begin{equation}
\lambda(\xi):=\sqrt{\abs{\bjj(\xi)}^{2}+m^{2}}. \label{lambda}
\end{equation}
Notice that $\bjj(\xi)$ and $\lambda(\xi)$ depend only on
$(\abs{\xi_{1}}_{p},\abs{\xi_{2}}_{p},\abs{\xi_{3}}_{p})$. By (\ref{Clifford}),
$(\balpha\cdot\bjj(\xi))^{2}=\abs{\bjj(\xi)}^{2}\Id$ and
$(\bsigma\cdot\bjj(\xi))^{2}=\abs{\bjj(\xi)}^{2}\boldsymbol{1}$. We also recall
that $k\cdot x=\sum_{i=1}^{3}k_{i}x_{i}\in\Qp$, for $k,x\in\Qp^{3}$.

In the next three sections, we will show that
\begin{equation}
\ii\frac{\partial}{\partial t}\Psi(t,x)=\bH_{0}\Psi(t,x),\qquad t\in\R,\ x\in\Qp^{3},
\label{Equ_dirac_2}
\end{equation}
with $\bH_{0}=\balpha\cdot\nabla+\beta m$ and $\nabla$ built from the
operators (\ref{J_x_l}), is a new family of $p$-adic Dirac equations.

%=====================================================================
\section{Plane waves}\label{Section_Plane_Waves}

\begin{definition}\label{Def_plane_waves}
Fix $k\in\Qp^{3}$ such that $\lambda(k)+m>0$ (this holds for every $k$ if
$m>0$; see Remark \ref{Remark_degenerate} for the case $\lambda(k)+m=0$). We
define the spinors
\begin{equation}
w^{(1)}(k)=\begin{bmatrix}\begin{bmatrix}1\\0\end{bmatrix}\\[8pt]\dfrac{\bsigma\cdot\bjj(k)}{\lambda(k)+m}\begin{bmatrix}1\\0\end{bmatrix}\end{bmatrix},\qquad
w^{(2)}(k)=\begin{bmatrix}\begin{bmatrix}0\\1\end{bmatrix}\\[8pt]\dfrac{\bsigma\cdot\bjj(k)}{\lambda(k)+m}\begin{bmatrix}0\\1\end{bmatrix}\end{bmatrix},
\label{Equ_3}
\end{equation}
\begin{equation}
w^{(3)}(k)=\begin{bmatrix}\dfrac{-\bsigma\cdot\bjj(k)}{\lambda(k)+m}\begin{bmatrix}0\\1\end{bmatrix}\\[8pt]\begin{bmatrix}0\\1\end{bmatrix}\end{bmatrix},\qquad
w^{(4)}(k)=\begin{bmatrix}\dfrac{-\bsigma\cdot\bjj(k)}{\lambda(k)+m}\begin{bmatrix}1\\0\end{bmatrix}\\[8pt]\begin{bmatrix}1\\0\end{bmatrix}\end{bmatrix},
\label{Equ_4}
\end{equation}
and the energies $E_{1}(k)=E_{2}(k)=\lambda(k)$, $E_{3}(k)=E_{4}(k)=-\lambda(k)$.
By a plane wave, we mean a function of the form
\begin{equation}
\Psi_{s}(t,x)=e^{-\ii E_{s}(k)t}\chi_{p}(k\cdot x)\,w^{(s)}(k),\qquad x\in\Qp^{3},\ t\in\R,\ s=1,2,3,4.
\label{Equ_1}
\end{equation}
Thus $\Psi_{1},\Psi_{2}$ are proportional to $e^{-\ii\lambda(k)t}$ (positive
energy), while $\Psi_{3},\Psi_{4}$ are proportional to $e^{+\ii\lambda(k)t}$
(negative energy).
\end{definition}

The operator $\bH_{0}$ is defined on bounded functions by the integral
formula (\ref{J_x_l}); in particular, $\bH_{0}\Psi_{s}(t,\cdot)$ makes sense
although $\Psi_{s}(t,\cdot)\notin L^{2}(\Qp^{3})\otimes\C^{4}$. As in the
classical case, plane waves are generalized eigenfunctions; superpositions
$\int\widehat{\psi}(k)\Psi_{s}(t,x)\,d^{3}k$ with
$\widehat{\psi}\in L^{2}(\Qp^{3})$ give genuine solutions; see Section
\ref{Section_3}.

\begin{proposition}\label{Proposition_1}
The plane waves (\ref{Equ_1}) are solutions of the $p$-adic Dirac equation
(\ref{Equ_dirac_2}). More precisely, setting
\[
h(k):=\balpha\cdot\bjj(k)+\beta m=\begin{bmatrix}m\boldsymbol{1}&\bsigma\cdot\bjj(k)\\\bsigma\cdot\bjj(k)&-m\boldsymbol{1}\end{bmatrix},
\]
we have $h(k)w^{(s)}(k)=E_{s}(k)w^{(s)}(k)$ for $s=1,2,3,4$. Moreover, the
vectors $w^{(s)}(k)$ are pairwise orthogonal in $\C^{4}$, and
$\abs{w^{(s)}(k)}^{2}=2\lambda(k)/(\lambda(k)+m)$.
\end{proposition}

\begin{proof}
By (\ref{Formula 1}), applied in each variable $x_{l}$,
$\bJ_{x_{l}}\chi_{p}(k\cdot x)=\jj_{l}(k_{l})\chi_{p}(k\cdot x)$, since
$\chi_{p}(k\cdot x)=\prod_{l}\chi_{p}(k_{l}x_{l})$. Hence
$\bH_{0}\left(\chi_{p}(k\cdot x)w\right)=\chi_{p}(k\cdot x)h(k)w$ for every
$w\in\C^{4}$, and, since $\ii\frac{\partial}{\partial t}e^{-\ii E t}=Ee^{-\ii
Et}$, the function (\ref{Equ_1}) solves (\ref{Equ_dirac_2}) if and only if
\begin{equation}
E_{s}(k)w^{(s)}(k)=h(k)w^{(s)}(k).\label{Eigenvalue_problem_1}
\end{equation}
Explicitly, writing $\jj_{l}=\jj_{l}(k_{l})$, (\ref{Eigenvalue_problem_1}) is
the linear system $(h(k)-E\Id)w=0$, where
\[
h(k)-E\Id=\begin{bmatrix}
m-E&0&\jj_{3}&\jj_{1}-\ii\jj_{2}\\
0&m-E&\jj_{1}+\ii\jj_{2}&-\jj_{3}\\
\jj_{3}&\jj_{1}-\ii\jj_{2}&-m-E&0\\
\jj_{1}+\ii\jj_{2}&-\jj_{3}&0&-m-E\end{bmatrix},
\]
whose determinant is $(m^{2}+\abs{\bjj(k)}^{2}-E^{2})^{2}$; thus non-trivial
solutions exist if and only if $E=\pm\lambda(k)$. We verify
(\ref{Eigenvalue_problem_1}) directly. Write $S:=\bsigma\cdot\bjj(k)$, so that
$S^{2}=\abs{\bjj(k)}^{2}\boldsymbol{1}=(\lambda^{2}-m^{2})\boldsymbol{1}$, with
$\lambda=\lambda(k)$. For $\eta\in\C^{2}$ and
$w=\left[\eta,\ \frac{S\eta}{\lambda+m}\right]^{T}$,
\[
h(k)w=\begin{bmatrix}m\eta+\frac{S^{2}\eta}{\lambda+m}\\ S\eta-\frac{mS\eta}{\lambda+m}\end{bmatrix}
=\begin{bmatrix}(m+\lambda-m)\eta\\ \frac{\lambda}{\lambda+m}S\eta\end{bmatrix}=\lambda w,
\]
which gives the cases $s=1,2$. For
$w=\left[-\frac{S\eta}{\lambda+m},\ \eta\right]^{T}$,
\[
h(k)w=\begin{bmatrix}-\frac{mS\eta}{\lambda+m}+S\eta\\ -\frac{S^{2}\eta}{\lambda+m}-m\eta\end{bmatrix}
=\begin{bmatrix}\frac{\lambda}{\lambda+m}S\eta\\ -\lambda\eta\end{bmatrix}=-\lambda w,
\]
which gives the cases $s=3,4$. Finally, $S$ is Hermitian, so
$\abs{w^{(s)}}^{2}=1+\frac{\abs{S\eta}^{2}}{(\lambda+m)^{2}}=1+\frac{\lambda-m}{\lambda+m}=\frac{2\lambda}{\lambda+m}$
for $\abs{\eta}=1$; $w^{(1)}\perp w^{(2)}$ and $w^{(3)}\perp w^{(4)}$ by a
direct computation, and the eigenvectors of the Hermitian matrix $h(k)$ for the
distinct eigenvalues $\pm\lambda$ are orthogonal (if $\lambda=0$ then
$\lambda+m=0$, which is excluded). The computation is the same as in the
classical case; see, e.g., \cite{Bjorken,Greiner}.
\end{proof}

\begin{remark}\label{Remark_degenerate}
The denominators in (\ref{Equ_3})--(\ref{Equ_4}) vanish if and only if
$\lambda(k)+m=0$, i.e., if and only if $m=0$ and $\bjj(k)=0$. In this case
$h(k)=0$, and every function $\chi_{p}(k\cdot x)w$, $w\in\C^{4}$, is a
time-independent (zero-energy) solution of (\ref{Equ_dirac_2}). This happens,
for instance, for $k=0$, and, by Example \ref{Example_kernel}, for every $k$ in
the ball $\{\abs{k_{l}}_{p}\leq p^{-L_{l}},\ l=1,2,3\}$ if each $J_{l}$ is
supported in $p^{-L_{l}}\Zp$. These zero-energy modes have no analogue for the
massive standard Dirac equation, and they are discussed in Section
\ref{Section_zero_modes}. If $m>0$, then $\lambda(k)+m\geq2m>0$ for every
$k$, and no singularity occurs. The normalized spinors are
$\sqrt{(\lambda(k)+m)/(2\lambda(k))}\,w^{(s)}(k)$.
\end{remark}

%=====================================================================
\section{The free $p$-adic Dirac operator and its spectrum}\label{Section_3}

\subsection{Some function spaces}\label{Section_Some_Function_Spaces}

This section uses a notation similar to the one used in \cite[Chapter
1]{Thaller} to compare the standard and the $p$-adic Dirac operators quickly.
Furthermore, we use several results and calculations in \cite[Chapter
1]{Thaller}. We set
\[
\HH:=L^{2}(\Qp^{3})\otimes\C^{4},
\]
and identify $\HH$, via the isomorphism of complex Hilbert spaces
$\phi\otimes e_{j}\mapsto(0,\ldots,\phi,\ldots,0)$ ($\phi$ in the $j$-th
place), with the orthogonal direct sum $\bigoplus_{j=1}^{4}L^{2}(\Qp^{3})=L^{2}(\Qp^{3})^{4}$.
Thus we identify the elements of $\HH$ with column vectors of the form
$\psi(x)=[\psi_{1}(x),\ldots,\psi_{4}(x)]^{T}$, $x\in\Qp^{3}$.
The inner product is given by
\[
(\psi,\phi)=\int_{\Qp^{3}}\sum_{l=1}^{4}\psi_{l}(x)\overline{\phi_{l}(x)}\,d^{3}x=\sum_{l=1}^{4}\langle\psi_{l},\phi_{l}\rangle,
\]
and the norm by
$\norm{\psi}=\sqrt{\sum_{l=1}^{4}\norm{\psi_{l}}_{2}^{2}}$.
Given $\psi\in\HH$ with integrable components, its Fourier transform is
defined as $(\F\psi)(\xi)=\widehat{\psi}(\xi)=[\widehat{\psi}_{1}(\xi),\ldots,\widehat{\psi}_{4}(\xi)]^{T}$, where
\[
(\F\psi_{i})(\xi)=\widehat{\psi_{i}}(\xi)=\int_{\Qp^{3}}\chi_{p}(\xi\cdot x)\psi_{i}(x)\,d^{3}x,\qquad\xi\in\Qp^{3},
\]
for $i=1,2,3,4$. The Fourier transform extends to a uniquely defined unitary
operator (denoted as $\F$) on the Hilbert space $\HH$.

\begin{proposition}\label{Proposition_0}
The mapping $\psi\mapsto\bH_{0}\psi$ is a bounded, self-adjoint, linear
operator on $\HH$, with
$\norm{\bH_{0}}\leq M:=\left(m^{2}+4\sum_{l=1}^{3}\norm{J_{l}}_{1}^{2}\right)^{1/2}$.
Furthermore, $\{e^{-\ii t\bH_{0}}\}_{t\in\R}$ is a norm-continuous group
of unitary operators on $\HH$, and, for every $\psi_{0}\in\HH$, the Cauchy
problem
\[
\left\{
\begin{array}{ll}
\ii\frac{\partial}{\partial t}\Psi(t,x)=\bH_{0}\Psi(t,x), & t\in\R,\ x\in\Qp^{3},\\[4pt]
\Psi(0,x)=\psi_{0}(x), &
\end{array}
\right.
\]
has a unique solution $\Psi\in C^{1}(\R,\HH)$, given by
$\Psi(t,x)=e^{-\ii t\bH_{0}}\psi_{0}(x)$. In particular, the
restriction $\{e^{-\ii t\bH_{0}}\}_{t\geq0}$, which describes the evolution
forward in time, is a unitary semigroup.
\end{proposition}

\begin{proof}
Each $\bJ_{x_{l}}$ is a bounded self-adjoint operator on
$L^{2}(\Qp^{3})$, being unitarily equivalent, via $\F$, to the multiplication
by the bounded real function $\jj_{l}(\xi_{l})$; and each $\alpha_{l}$, and
$\beta$, is a Hermitian matrix. Hence each term $\alpha_{l}\bJ_{x_{l}}$ (i.e.,
$\bJ_{x_{l}}\otimes\alpha_{l}$) and $\beta m$ is bounded and self-adjoint on
$\HH$, and so is $\bH_{0}$. The bound on $\norm{\bH_{0}}$ follows from the
description of $\bH_{0}$ in Fourier space given in Section
\ref{Section_Fourier_space}: $\norm{\bH_{0}}=\sup_{\xi}\lambda(\xi)$ and
$\abs{\jj_{l}}\leq2\norm{J_{l}}_{1}$. Since $\bH_{0}$ is bounded,
$e^{-\ii t\bH_{0}}=\sum_{k\geq0}(-\ii t\bH_{0})^{k}/k!$ converges in norm, it
is unitary because $\bH_{0}$ is self-adjoint, $t\mapsto e^{-\ii t\bH_{0}}$ is
a norm-continuous group, and $\Psi(t)=e^{-\ii t\bH_{0}}\psi_{0}$ is a
$C^{1}$ solution. If $\Phi\in C^{1}(\R,\HH)$ is another solution, then
$\frac{d}{dt}e^{\ii t\bH_{0}}\Phi(t)=0$, so $\Phi(t)=e^{-\ii t\bH_{0}}\psi_{0}$.
\end{proof}

\subsection{The free Dirac operator in the Fourier space}\label{Section_Fourier_space}

The results presented in this section are analogs of the results for the
standard Dirac operator. In particular, the calculations used here are the
same as the ones given in \cite[Section 1.4.1]{Thaller}. The Hamiltonian
$\bH_{0}$ is a matrix pseudo-differential operator on $\HH$, which is
transformed via $\F$ into a matrix multiplication operator:
\begin{equation}
(\bH_{0}\phi)(x)=\F_{\xi\rightarrow x}^{-1}\left(h(\xi)\F_{x\rightarrow\xi}\phi\right),\qquad\text{for }\phi\in\HH,
\label{pseudodifferetial_op_1}
\end{equation}
where
\[
h(\xi):=\balpha\cdot\bjj(\xi)+\beta m=\begin{bmatrix}m\boldsymbol{1}&\bsigma\cdot\bjj(\xi)\\\bsigma\cdot\bjj(\xi)&-m\boldsymbol{1}\end{bmatrix}.
\]
The matrix $h(\xi)$ is a $4\times4$ Hermitian matrix satisfying
$h(\xi)^{2}=\lambda(\xi)^{2}\Id$, by (\ref{Clifford}). It has the
eigenvalues
$\lambda_{1}(\xi)=\lambda_{2}(\xi)=-\lambda_{3}(\xi)=-\lambda_{4}(\xi)=\lambda(\xi)$,
with $\lambda(\xi)$ as in (\ref{lambda}).

We define the $4\times4$ matrix
\begin{equation}
u(\xi):=\frac{(m+\lambda(\xi))\Id+\beta\,\balpha\cdot\bjj(\xi)}{\sqrt{2\lambda(\xi)(m+\lambda(\xi))}}\quad\text{if }\lambda(\xi)>0,\qquad
u(\xi):=\Id\quad\text{if }\lambda(\xi)=0.
\label{u_matrix}
\end{equation}
Notice that $\lambda(\xi)=0$ if and only if $m=0$ and $\bjj(\xi)=0$; in this
case $h(\xi)=0$. When $\bjj(\xi)\neq0$, $u(\xi)$ can be rewritten as
\[
u(\xi)=a_{+}(\xi)\Id+a_{-}(\xi)\beta\frac{\balpha\cdot\bjj(\xi)}{\abs{\bjj(\xi)}},\qquad
a_{\pm}(\xi)=\frac{1}{\sqrt{2}}\sqrt{1\pm\frac{m}{\lambda(\xi)}},
\]
while $u(\xi)=\Id$ when $\bjj(\xi)=0$ (because then $\lambda(\xi)=m$).

\begin{lemma}\label{Lemma_u}
For every $\xi\in\Qp^{3}$, $u(\xi)$ is a unitary matrix,
\[
u(\xi)^{-1}=u(\xi)^{\ast}=\frac{(m+\lambda(\xi))\Id-\beta\,\balpha\cdot\bjj(\xi)}{\sqrt{2\lambda(\xi)(m+\lambda(\xi))}}\quad(\lambda(\xi)>0),
\]
and
\begin{equation}
u(\xi)h(\xi)u(\xi)^{-1}=\beta\lambda(\xi).\label{diagonalization}
\end{equation}
\end{lemma}

\begin{proof}
If $\lambda(\xi)=0$, then $h(\xi)=0$ and $u(\xi)=\Id$, and there is
nothing to prove. Assume $\lambda=\lambda(\xi)>0$, write $c=m+\lambda>0$,
$N=\sqrt{2\lambda c}$, and $B=\beta\,\balpha\cdot\bjj(\xi)$. By
(\ref{Clifford}) and since $\bjj(\xi)$ is real,
$B^{\ast}=(\balpha\cdot\bjj)\beta=-\beta\,\balpha\cdot\bjj=-B$ and
$B^{2}=-\beta^{2}(\balpha\cdot\bjj)^{2}=-\abs{\bjj}^{2}\Id$. Hence
$u=(c\Id+B)/N$, $u^{\ast}=(c\Id-B)/N$, and
\[
uu^{\ast}=u^{\ast}u=\frac{c^{2}\Id-B^{2}}{N^{2}}=\frac{(m+\lambda)^{2}+\lambda^{2}-m^{2}}{2\lambda(m+\lambda)}\Id=\Id.
\]
Next, using $\beta\,\balpha\cdot\bjj\,\beta=-\balpha\cdot\bjj$,
\[
(c\Id+B)h=c\,\balpha\cdot\bjj+cm\beta+\beta\abs{\bjj}^{2}-m\,\balpha\cdot\bjj
=\lambda\,\balpha\cdot\bjj+\lambda(m+\lambda)\beta=\beta\lambda\,(c\Id+B),
\]
where we used $c-m=\lambda$ and $cm+\abs{\bjj}^{2}=m^{2}+m\lambda+\lambda^{2}-m^{2}$.
Thus $uh=\beta\lambda u$, which is (\ref{diagonalization}).
\end{proof}

By using (\ref{pseudodifferetial_op_1}) and (\ref{diagonalization}), the
unitary transformation
\[
\mathcal{W}:=u\F:\HH\rightarrow\HH
\]
(where $u$ acts as the multiplication operator by the measurable,
unitary-matrix-valued function $u(\xi)$) converts the $p$-adic Dirac
operator $\bH_{0}$ into a multiplication operator by the diagonal matrix
$\beta\lambda(\xi)$,
\begin{equation}
\bH_{0}=\mathcal{W}^{-1}\beta\lambda(\xi)\mathcal{W}\quad\text{in }\HH.
\label{diagonalization_2}
\end{equation}
Indeed,
$\mathcal{W}\bH_{0}\mathcal{W}^{-1}=u\F\bH_{0}\F^{-1}u^{-1}=uhu^{-1}=\beta\lambda$.

\subsection{The spectrum of $\bH_{0}$}

Since each $\jj_{l}(\xi_{l})$ depends only on $\abs{\xi_{l}}_{p}$, the
function $\lambda$ is constant on each of the sets
\[
R_{\boldsymbol{\gamma}}:=\left\{\xi\in\Qp^{3};\ \abs{\xi_{l}}_{p}=p^{\gamma_{l}},\ l=1,2,3\right\},\qquad
\boldsymbol{\gamma}=(\gamma_{1},\gamma_{2},\gamma_{3})\in\mathbb{Z}^{3},
\]
which are open, compact, pairwise disjoint, have positive Haar measure, and
cover $\Qp^{3}$ up to the null set $\{\xi;\ \xi_{1}\xi_{2}\xi_{3}=0\}$. We
denote by $\lambda_{\boldsymbol{\gamma}}$ the value of $\lambda$ on
$R_{\boldsymbol{\gamma}}$, and set
$\Lambda:=\{\lambda_{\boldsymbol{\gamma}};\ \boldsymbol{\gamma}\in\mathbb{Z}^{3}\}$,
a countable subset of $[m,M]$, with $M$ as in Proposition \ref{Proposition_0}.

\begin{theorem}\label{Theorem_spectrum}
The free Dirac operator $\bH_{0}$ is self-adjoint on $\HH$, and:
\begin{enumerate}[label=(\roman*)]
\item its spectrum $\sigma(\bH_{0})$ is the union of the essential ranges of
the functions $\pm\lambda(\xi):\Qp^{3}\rightarrow\R$, which equals the closure
of $\Lambda\cup(-\Lambda)$ in $\R$; in particular,
$\sigma(\bH_{0})\subset[-M,-m]\cup[m,M]$, and $\pm m\in\sigma(\bH_{0})$;
\item the spectrum of $\bH_{0}$ is pure point: each $\pm\lambda_{\boldsymbol{\gamma}}$
is an eigenvalue of infinite multiplicity, and $\HH$ is the orthogonal direct
sum of the corresponding eigenspaces; in particular,
$\sigma(\bH_{0})=\sigma_{\mathrm{ess}}(\bH_{0})$.
\end{enumerate}
\end{theorem}

\begin{proof}
By (\ref{diagonalization_2}), $\bH_{0}$ is unitarily equivalent to the
multiplication operator by $\beta\lambda(\xi)$ on $\HH$, i.e., to the direct
sum of two copies of the multiplication by $\lambda(\xi)$ and two copies of
the multiplication by $-\lambda(\xi)$ on $L^{2}(\Qp^{3})$. The spectrum of a
multiplication operator is the essential range of the multiplier, see
\cite[Section VII.2]{Reed-Simon-I}, which gives the first assertion in (i).
Since each $R_{\boldsymbol{\gamma}}$ has positive measure, every
$\lambda_{\boldsymbol{\gamma}}$ belongs to the essential range of $\lambda$;
conversely, $\lambda$ takes values in $\Lambda$ outside a null set, so its
essential range is contained in $\overline{\Lambda}$; being closed, it equals
$\overline{\Lambda}$. The bound $\abs{\jj_{l}}\leq2\norm{J_{l}}_{1}$ gives
$m\leq\lambda\leq M$. Since $\jj_{l}$ is continuous with $\jj_{l}(0)=0$,
$\lambda(\xi)\rightarrow m$ as $\xi\rightarrow0$, hence
$m\in\overline{\Lambda}$. For (ii), the space
$E_{\boldsymbol{\gamma}}^{\pm}:=\mathcal{W}^{-1}\left(L^{2}(R_{\boldsymbol{\gamma}})\otimes\tfrac{1}{2}(\Id\pm\beta)\C^{4}\right)$
is infinite dimensional (because $R_{\boldsymbol{\gamma}}$ is open), and
$\bH_{0}\psi=\pm\lambda_{\boldsymbol{\gamma}}\psi$ for $\psi\in E_{\boldsymbol{\gamma}}^{\pm}$
by (\ref{diagonalization_2}). Since the sets $R_{\boldsymbol{\gamma}}$ cover
$\Qp^{3}$ up to a null set,
$\HH=\bigoplus_{\boldsymbol{\gamma},\pm}E_{\boldsymbol{\gamma}}^{\pm}$
(eigenspaces corresponding to the same value are merged). An eigenvalue of
infinite multiplicity belongs to the essential spectrum, and the essential
spectrum is closed, so $\sigma_{\mathrm{ess}}(\bH_{0})=\sigma(\bH_{0})$.
\end{proof}

\begin{remark}
We denote by $\sigma(\bH_{0}^{Arch})$ the spectrum of the standard free Dirac
operator, with the normalization $c=1$; by \cite[Theorem 1.1]{Thaller},
$\sigma(\bH_{0}^{Arch})=(-\infty,-m]\cup[m,\infty)$. Then
$\sigma(\bH_{0})\subset\sigma(\bH_{0}^{Arch})$. However, in contrast with the
standard case, where the spectrum is purely absolutely continuous, the
spectrum of $\bH_{0}$ is bounded and pure point. This is a consequence of the
ultrametric structure of $\Qp$: radial functions are locally constant away
from the origin.
\end{remark}

\subsection{Positive-, negative-, and zero-energy subspaces}\label{Section_zero_modes}

In the Hilbert space $\mathcal{W}\HH$, the $p$-adic Dirac operator is
diagonal, see (\ref{diagonalization_2}). Following \cite[Section
1.4.2]{Thaller}, and in order to include the massless case $m=0$, we define
the positive-, negative-, and zero-energy subspaces of $\bH_{0}$ by means of
the spectral theorem (functional calculus, \cite[Section VII.2]{Reed-Simon-I}):
\[
\mathcal{P}_{pos}:=1_{(0,\infty)}(\bH_{0}),\qquad
\mathcal{P}_{neg}:=1_{(-\infty,0)}(\bH_{0}),\qquad
\mathcal{P}_{0}:=1_{\{0\}}(\bH_{0}),
\]
and $\HH_{pos}:=\mathcal{P}_{pos}\HH$, $\HH_{neg}:=\mathcal{P}_{neg}\HH$,
$\HH_{0}^{\ker}:=\mathcal{P}_{0}\HH=\ker\bH_{0}$. These are mutually
orthogonal projections with $\mathcal{P}_{pos}+\mathcal{P}_{neg}+\mathcal{P}_{0}=1$
and they commute with $\bH_{0}$, so
\begin{equation}
\HH=\HH_{pos}\oplus\HH_{0}^{\ker}\oplus\HH_{neg},\label{partition_1}
\end{equation}
and each summand is invariant under $\bH_{0}$ and under $e^{-\ii t\bH_{0}}$.
By (\ref{diagonalization_2}), with $Z_{0}:=\{\xi\in\Qp^{3};\ \lambda(\xi)=0\}$,
\begin{equation}
\mathcal{P}_{\substack{pos\\neg}}=\mathcal{W}^{-1}\,\tfrac{1}{2}(\Id\pm\beta)\,1_{\Qp^{3}\smallsetminus Z_{0}}(\xi)\,\mathcal{W},\qquad
\mathcal{P}_{0}=\F^{-1}1_{Z_{0}}(\xi)\F,\label{Opertor_P}
\end{equation}
where in the last formula we used $u(\xi)=\Id$ on $Z_{0}$.
For $\psi\in\HH_{pos}$, $\phi_{+}:=\mathcal{W}\psi$ satisfies
$(\psi,\bH_{0}\psi)=(\phi_{+},\lambda(\xi)\phi_{+})\geq0$, with equality only
if $\psi=0$; similarly $(\psi,\bH_{0}\psi)<0$ for $0\neq\psi\in\HH_{neg}$.

Let $\abs{\bH_{0}}:=\sqrt{\bH_{0}^{2}}$, which is the pseudo-differential
operator with symbol $\lambda(\xi)\Id$, i.e.,
$\abs{\bH_{0}}=\sqrt{\bJ_{x_{1}}^{2}+\bJ_{x_{2}}^{2}+\bJ_{x_{3}}^{2}+m^{2}}$,
and $\operatorname{sgn}\bH_{0}:=\mathcal{P}_{pos}-\mathcal{P}_{neg}$, which is
the pseudo-differential operator with symbol
$\lambda(\xi)^{-1}h(\xi)$ on $\Qp^{3}\smallsetminus Z_{0}$ and $0$ on
$Z_{0}$. Then $\bH_{0}=\abs{\bH_{0}}\operatorname{sgn}\bH_{0}$, which is the
polar decomposition of $\bH_{0}$, and $\bH_{0}\psi=\pm\abs{\bH_{0}}\psi$ for
$\psi\in\HH_{\substack{pos\\neg}}$.

\medskip\noindent\textbf{The massive case.} If $m>0$, then $\lambda\geq m>0$,
$Z_{0}=\emptyset$, $\mathcal{P}_{0}=0$, $\abs{\bH_{0}}\geq m$ is invertible,
$\operatorname{sgn}\bH_{0}=\bH_{0}\abs{\bH_{0}}^{-1}$, and we recover the
familiar formulas
$\mathcal{P}_{\substack{pos\\neg}}=\frac{1}{2}\left(1\pm\bH_{0}\abs{\bH_{0}}^{-1}\right)$
and $\HH=\HH_{pos}\oplus\HH_{neg}$.

\medskip\noindent\textbf{The massless case.} If $m=0$, then $\lambda(\xi)=\abs{\bjj(\xi)}$,
$0\in\sigma(\bH_{0})$ by Theorem \ref{Theorem_spectrum}, and the operator
$\abs{\bH_{0}}$ is not invertible (its inverse would be unbounded, because
$\lambda(\xi)\rightarrow0$ as $\xi\rightarrow0$). The expression
$\bH_{0}/\abs{\bH_{0}}$ is therefore not defined as a bounded operator, and
the decomposition (\ref{partition_1}) must include the kernel of $\bH_{0}$.

\begin{proposition}\label{Prop_zero_modes}
Assume $m=0$, and set $Z^{(l)}:=\{\xi_{l}\in\Qp;\ \jj_{l}(\xi_{l})=0\}$, so
that $Z_{0}=Z^{(1)}\times Z^{(2)}\times Z^{(3)}$. Then
\[
\ker\bH_{0}=\left\{\psi\in\HH;\ \widehat{\psi}(\xi)=0\text{ for almost every }\xi\notin Z_{0}\right\}.
\]
Moreover:
\begin{enumerate}[label=(\roman*)]
\item if every $Z^{(l)}$ has positive Haar measure, then $\ker\bH_{0}$ is
infinite dimensional, and $0$ is an eigenvalue of infinite multiplicity of
$\bH_{0}$. This is the case if each $J_{l}$ is supported in a ball
$p^{-L_{l}}\Zp$: then $Z_{0}\supset p^{L_{1}}\Zp\times p^{L_{2}}\Zp\times p^{L_{3}}\Zp$,
and $\ker\bH_{0}$ contains every $\psi\in\HH$ that is invariant under the
translations by $p^{-L_{1}}\Zp\times p^{-L_{2}}\Zp\times p^{-L_{3}}\Zp$;
\item if some $J_{l}$ satisfies $J_{l}\geq0$ and $J_{l}(p^{k})>0$ for
infinitely many $k>0$, then $Z^{(l)}=\{0\}$, $\ker\bH_{0}=\{0\}$, but
$0\in\sigma(\bH_{0})$ is an accumulation point of eigenvalues.
\end{enumerate}
\end{proposition}

\begin{proof}
Since $h(\xi)^{2}=\lambda(\xi)^{2}\Id$, $h(\xi)$ is invertible for
$\xi\notin Z_{0}$ and vanishes for $\xi\in Z_{0}$; hence
$\bH_{0}\psi=0$ if and only if $h(\xi)\widehat{\psi}(\xi)=0$ a.e., if and only
if $\widehat{\psi}=0$ a.e. outside $Z_{0}$. (i) If $Z_{0}$ has positive
measure, it contains a set $R_{\boldsymbol{\gamma}}$ (because $\lambda$ is
constant on these sets), and $\F^{-1}(L^{2}(R_{\boldsymbol{\gamma}})\otimes\C^{4})$
is infinite dimensional. If $J_{l}$ is supported in $p^{-L_{l}}\Zp$, then
$\jj_{l}(\xi_{l})=0$ for $\abs{\xi_{l}}_{p}\leq p^{-L_{l}}$ by Example
\ref{Example_kernel}. A function $\psi\in L^{2}$ is invariant under
translations by a compact open subgroup $G$ if and only if $\widehat{\psi}$ is
supported in the annihilator $\{\xi;\ \chi_{p}(\xi\cdot a)=1\ \forall a\in G\}$,
which for $G=\prod_{l}p^{-L_{l}}\Zp$ is $\prod_{l}p^{L_{l}}\Zp$. (ii) For
$\xi_{l}\neq0$, $\widehat{J}_{l}(0)-\widehat{J}_{l}(\abs{\xi_{l}}_{p})=\int_{\Qp}J_{l}(\abs{y}_{p})\{1-\operatorname{Re}\chi_{p}(\xi_{l}y)\}dy$.
The integrand is non-negative. Choose $k>0$ with $J_{l}(p^{k})>0$ and
$p^{k}\abs{\xi_{l}}_{p}\geq p$; on the sphere $\abs{y}_{p}=p^{k}$ the
function $y\mapsto\chi_{p}(\xi_{l}y)$ is not identically $1$ (its integral
over the sphere is $0$ or $-p^{k-1}$, whereas the sphere has measure
$p^{k}(1-p^{-1})$), so the integral is strictly positive and
$\jj_{l}(\xi_{l})\neq0$. Thus $Z^{(l)}=\{0\}$ has measure zero, and
$\ker\bH_{0}=\{0\}$. The last assertion follows from Theorem
\ref{Theorem_spectrum}(i).
\end{proof}

\begin{remark}
In the setting of Section \ref{Section_CTQW_II} (kernels supported in $\Zp$),
the zero-energy subspace of the massless operator contains the functions on
$\Zp^{3}$ that are constant, tensored with $\C^{4}$; on the finite graph
$G_{n}^{3}$ this corresponds to the uniform distribution over the vertices.
\end{remark}

Again, following the classical case, we define the Foldy--Wouthuysen
transformation as
$\mathcal{U}_{FW}=\F^{-1}\mathcal{W}=\F^{-1}u\F$.
It transforms the free Dirac operator into the pseudo-differential operator
\begin{align*}
\mathcal{U}_{FW}\bH_{0}\mathcal{U}_{FW}^{-1}&=\F^{-1}\,u h u^{-1}\,\F\\
&=\begin{bmatrix}
\sqrt{\bJ_{x_{1}}^{2}+\bJ_{x_{2}}^{2}+\bJ_{x_{3}}^{2}+m^{2}}&\boldsymbol{0}\\
\boldsymbol{0}&-\sqrt{\bJ_{x_{1}}^{2}+\bJ_{x_{2}}^{2}+\bJ_{x_{3}}^{2}+m^{2}}
\end{bmatrix}=\beta\abs{\bH_{0}}.
\end{align*}
We interpret this formula as the fact that the free Dirac equation is
unitarily equivalent to a pair of (two-component) square-root Klein--Gordon
equations.

%=====================================================================
\section{Charge conjugation}\label{Section_4}

Following the standard case, see \cite[Section 1.4.6]{Thaller}, the $p$-adic
Dirac operator for a charge $e\in\R$ in an external electromagnetic field is
given by
\[
\bH(e)=\bH(e;\boldsymbol{A},\varphi):=\balpha\cdot\left(\nabla-e\boldsymbol{A}(t,x)\right)+\beta m+e\varphi(t,x)\Id,
\]
where the scalar potential $\varphi:\R\times\Qp^{3}\rightarrow\R$ and the
vector potential $\boldsymbol{A}=(A_{1},A_{2},A_{3}):\R\times\Qp^{3}\rightarrow\R^{3}$
are bounded measurable functions, acting as multiplication operators. Thus,
for every $t$, $\bH(e)$ is a bounded self-adjoint operator on $\HH$.

In the standard case, the charge conjugation is the antiunitary map
$\Psi\mapsto U\overline{\Psi}$ with $U=\ii\beta\alpha_{2}$ (up to a phase).
That choice relies on the identity $\overline{-\ii\partial_{x_{l}}\Psi}=\ii\partial_{x_{l}}\overline{\Psi}$.
In our setting, the kernels $J_{l}$ are real, so
$\overline{\bJ_{x_{l}}\Psi}=\bJ_{x_{l}}\overline{\Psi}$ (Remark
\ref{Remark_nature_J}), and the matrix $U$ has to be chosen differently. We
define the charge conjugation $\mathcal{C}$ as the antiunitary transformation
\[
\mathcal{C}\Psi:=U_{\mathcal{C}}\overline{\Psi},\qquad U_{\mathcal{C}}:=\alpha_{2}.
\]
Since $\alpha_{2}$ is unitary and $\overline{\alpha_{2}}=-\alpha_{2}$, we have
$\mathcal{C}^{2}\Psi=\alpha_{2}\overline{\alpha_{2}}\Psi=-\alpha_{2}^{2}\Psi=-\Psi$,
that is, $\mathcal{C}^{2}=-1$.

\begin{lemma}\label{Lemma_charge_conj}
With the above notation, the following assertions hold.
\begin{enumerate}[label=(\roman*)]
\item $\mathcal{C}\bJ_{x_{l}}=\bJ_{x_{l}}\mathcal{C}$,
$\mathcal{C}\alpha_{l}=-\alpha_{l}\mathcal{C}$ for $l=1,2,3$, and
$\mathcal{C}\beta=-\beta\mathcal{C}$, as operators on $\HH$. Moreover, up to
a phase factor, $U_{\mathcal{C}}=\alpha_{2}$ is the only $4\times4$ matrix $U$
such that $U\overline{\alpha_{l}}U^{-1}=-\alpha_{l}$, $l=1,2,3$, and
$U\beta U^{-1}=-\beta$.
\item $\mathcal{C}\bH_{0}\mathcal{C}^{-1}=-\bH_{0}$.
\item $\mathcal{C}\bH(e;\boldsymbol{A},\varphi)\mathcal{C}^{-1}=-\bH(-e;-\boldsymbol{A},\varphi)$.
In particular, in a purely electrostatic field ($\boldsymbol{A}=0$),
$\mathcal{C}\bH(e)\mathcal{C}^{-1}=-\bH(-e)$.
\item If $\Psi\in C^{1}(\R,\HH)$ is a solution of
\begin{equation}
\ii\frac{\partial}{\partial t}\Psi(t,x)=\bH(e;\boldsymbol{A},\varphi)\Psi(t,x),\label{Eq_15}
\end{equation}
then
\[
\ii\frac{\partial}{\partial t}\mathcal{C}\Psi(t,x)=\bH(-e;-\boldsymbol{A},\varphi)\,\mathcal{C}\Psi(t,x).
\]
\end{enumerate}
\end{lemma}

\begin{proof}
(i) Since $J_{l}$ is real valued, $\overline{\bJ_{x_{l}}\Psi}=\bJ_{x_{l}}\overline{\Psi}$,
and $\bJ_{x_{l}}$ acts componentwise, so it commutes with the matrix
$\alpha_{2}$; hence $\mathcal{C}\bJ_{x_{l}}=\bJ_{x_{l}}\mathcal{C}$. Next,
$\mathcal{C}\alpha_{l}\Psi=\alpha_{2}\overline{\alpha_{l}}\,\overline{\Psi}$.
By (\ref{Clifford}), $\alpha_{2}\overline{\alpha_{l}}=\alpha_{2}\alpha_{l}=-\alpha_{l}\alpha_{2}$
for $l=1,3$, and $\alpha_{2}\overline{\alpha_{2}}=-\alpha_{2}\alpha_{2}$; hence
$\mathcal{C}\alpha_{l}=-\alpha_{l}\mathcal{C}$ for $l=1,2,3$. Similarly
$\alpha_{2}\overline{\beta}=\alpha_{2}\beta=-\beta\alpha_{2}$. For the
uniqueness, the conditions on $U$ are equivalent to
$U\alpha_{1}=-\alpha_{1}U$, $U\alpha_{3}=-\alpha_{3}U$, $U\beta=-\beta U$,
$U\alpha_{2}=\alpha_{2}U$. The $16$ products
$\Gamma_{S}=\prod_{g\in S}g$, $S\subset\{\alpha_{1},\alpha_{2},\alpha_{3},\beta\}$
(in a fixed order), form a basis of the space of $4\times4$ complex
matrices, and, by (\ref{Clifford}), each generator $g$ satisfies
$g\Gamma_{S}=\pm\Gamma_{S}g$, with sign $(-1)^{\abs{S}-1}$ if $g\in S$ and
$(-1)^{\abs{S}}$ if $g\notin S$. Writing $U=\sum_{S}c_{S}\Gamma_{S}$, the four
conditions force $c_{S}=0$ unless $\Gamma_{S}$ anticommutes with
$\alpha_{1},\alpha_{3},\beta$ and commutes with $\alpha_{2}$; if $\abs{S}$ is
odd, this means $S=\{\alpha_{2}\}$, and if $\abs{S}$ is even, this means
$S=\{\alpha_{1},\alpha_{3},\beta\}$, which has odd cardinality, a
contradiction. Hence $U=c\,\alpha_{2}$.

(ii) By (i),
$\mathcal{C}\bH_{0}=\sum_{l}\mathcal{C}\alpha_{l}\bJ_{x_{l}}+m\mathcal{C}\beta=-\sum_{l}\alpha_{l}\bJ_{x_{l}}\mathcal{C}-m\beta\mathcal{C}=-\bH_{0}\mathcal{C}$.

(iii) The functions $A_{l}$, $\varphi$ are real valued, so the corresponding
multiplication operators commute with $\mathcal{C}$. By (i),
\begin{align*}
\mathcal{C}\bH(e;\boldsymbol{A},\varphi)\mathcal{C}^{-1}&=-\balpha\cdot\nabla+e\balpha\cdot\boldsymbol{A}-\beta m+e\varphi\Id\\
&=-\left\{\balpha\cdot\left(\nabla-(-e)(-\boldsymbol{A})\right)+\beta m+(-e)\varphi\Id\right\}
=-\bH(-e;-\boldsymbol{A},\varphi).
\end{align*}

(iv) Apply $\mathcal{C}$ to (\ref{Eq_15}). Since $\mathcal{C}$ is antilinear
and continuous, $\mathcal{C}(\ii\partial_{t}\Psi)=-\ii\partial_{t}\mathcal{C}\Psi$,
and by (iii),
$-\ii\partial_{t}\mathcal{C}\Psi=\mathcal{C}\bH(e;\boldsymbol{A},\varphi)\Psi=-\bH(-e;-\boldsymbol{A},\varphi)\mathcal{C}\Psi$.
\end{proof}

\begin{remark}\label{Remark_charge_conj}
(1) In the standard case, $\mathcal{C}\bH(e)\mathcal{C}^{-1}=-\bH(-e)$ holds
with the \emph{same} vector potential, and $\mathcal{C}^{2}=1$. Both
differences come from the fact that the momentum-type operators
$\bJ_{x_{l}}$ are invariant under complex conjugation, as are the
multiplication operators $A_{l}$: in the $p$-adic model, the kinetic term
$\balpha\cdot\nabla$ and the magnetic coupling $\balpha\cdot\boldsymbol{A}$
transform in the same way, so a symmetry that reverses the sign of the former
must also reverse the sign of the latter. In the presence of a vector
potential, the symmetry $\mathcal{C}$ thus combines charge conjugation with
the reversal $\boldsymbol{A}\mapsto-\boldsymbol{A}$ (as happens, in the
standard theory, under time reversal). By the uniqueness part of Lemma
\ref{Lemma_charge_conj}(i), this cannot be remedied by choosing a different
matrix $U_{\mathcal{C}}$ while keeping the operators $\bJ_{x_{l}}$.

(2) In \cite{Zuniga-Dirac-Causality}, charge conjugation was established for a
$p$-adic Dirac equation built from the Vladimirov operator; that result cannot
be transferred to the present setting, since the operators $\bJ_{x_{l}}$ are
different. In particular, the matrix $-\ii\beta\alpha_{2}$ does not implement a
charge-conjugation symmetry of $\bH_{0}$: one checks that
$(-\ii\beta\alpha_{2})\overline{\bH_{0}}(-\ii\beta\alpha_{2})^{-1}=\balpha\cdot\nabla-\beta m\neq-\bH_{0}$.
\end{remark}

\begin{proposition}\label{Prop_C_pos_neg}
The charge conjugation $\mathcal{C}$ satisfies
$\mathcal{C}\mathcal{P}_{pos}\mathcal{C}^{-1}=\mathcal{P}_{neg}$,
$\mathcal{C}\mathcal{P}_{neg}\mathcal{C}^{-1}=\mathcal{P}_{pos}$, and
$\mathcal{C}\mathcal{P}_{0}\mathcal{C}^{-1}=\mathcal{P}_{0}$. Consequently,
$\mathcal{C}$ maps $\HH_{neg}$ onto $\HH_{pos}$, $\HH_{pos}$ onto $\HH_{neg}$,
and $\ker\bH_{0}$ onto itself. Moreover, $\abs{\mathcal{C}\psi(x)}^{2}=\abs{\psi(x)}^{2}$
for every $\psi\in\HH$ and almost every $x\in\Qp^{3}$.
\end{proposition}

\begin{proof}
By Lemma \ref{Lemma_charge_conj}(ii), $\mathcal{C}\bH_{0}^{k}\mathcal{C}^{-1}=(-\bH_{0})^{k}$
for every $k$, so, for every polynomial $P$ with real coefficients,
$\mathcal{C}P(\bH_{0})\mathcal{C}^{-1}=P(-\bH_{0})$. By the Weierstrass
theorem on the compact set $\sigma(\bH_{0})\cup(-\sigma(\bH_{0}))$ and the
continuity of the functional calculus, $\mathcal{C}f(\bH_{0})\mathcal{C}^{-1}=f(-\bH_{0})$
for every real continuous function $f$, and then, by monotone (strong)
limits, for every real bounded Borel function $f$. Taking $f=1_{(0,\infty)}$,
$1_{(-\infty,0)}$, $1_{\{0\}}$ gives the first assertions. The last one
follows from the fact that $\alpha_{2}$ is unitary.
\end{proof}

Hence, in the free case and in the electrostatic case, the negative-energy
subspace of $\bH(e)$ is connected via a symmetry transformation with the
positive-energy subspace of the Dirac operator $\bH(-e)$ for a particle with
opposite charge (antiparticle). For $\mathcal{C}\psi$ in the positive-energy
subspace of $\bH(-e)$, by interpreting $\abs{\mathcal{C}\psi(x)}^{2}$ as a
position probability density, the equality
$\abs{\mathcal{C}\psi(x)}^{2}=\abs{\psi(x)}^{2}$ shows that the motion of a
negative-energy state $\psi$ is indistinguishable from that of a
positive-energy antiparticle. Then, as in the standard case, one obtains the
interpretation:
\[
\text{a state }\psi\in\HH_{neg}\text{ describes an antiparticle with positive energy.}
\]
In the presence of a vector potential, by Lemma \ref{Lemma_charge_conj}(iii),
the same interpretation holds provided the antiparticle is subject to the
reversed vector potential $-\boldsymbol{A}$.

%=====================================================================
\section{Continuous-time quantum walks I}\label{Section_CTQW_I}

In this section, we discretize the $p$-adic Dirac evolution on a countable
graph. We first describe what can be said for an arbitrary bounded
self-adjoint Hamiltonian (Section \ref{Section_coarse}), and then show that,
for $\bH_{0}$ and arbitrary integrable radial kernels $J_{1},J_{2},J_{3}$, one
obtains a CTQW in the strict sense (Section \ref{Section_exact_CTQW_I}). To
avoid a conflict with the index $l$ of the spatial directions, we use $n$ to
denote the discretization level.

\subsection{The countable graph $\widetilde{G}_{n}^{3}$ and the basis
$\psi_{\bj,b}$}

For an integer $n\geq1$, set
\[
\widetilde{G}_{n}:=\Qp/p^{n}\Zp,\qquad
\widetilde{G}_{n}^{3}:=\Qp^{3}/p^{n}\Zp^{3}=\left(\Qp/p^{n}\Zp\right)^{3}.
\]
We use a set of representatives of $\widetilde{G}_{n}$ of the form
$i=i_{-r}p^{-r}+i_{-r+1}p^{-r+1}+\ldots+i_{0}+\ldots+i_{n-1}p^{n-1}$, with
$r\geq0$ and digits $i_{k}\in\{0,1,\ldots,p-1\}$. The elements of
$\widetilde{G}_{n}^{3}$ are denoted as $\bi=(i_{1},i_{2},i_{3})$. Notice that
$\widetilde{G}_{n}^{3}$ is a countable set, and that, for $i\neq k$ in
$\widetilde{G}_{n}$, the number $\abs{i-k}_{p}>p^{-n}$ does not depend on the
choice of representatives. We consider the partition
\begin{equation}
\Qp^{3}=\bigsqcup_{\bi\in\widetilde{G}_{n}^{3}}\left(\bi+p^{n}\Zp^{3}\right),
\label{Partition_A}
\end{equation}
where the balls $\bi+p^{n}\Zp^{3}$ are pairwise disjoint. We denote by
$\Om(p^{n}\norm{x-\bi}_{p})$ the characteristic function of the ball
$\bi+p^{n}\Zp^{3}$, by $e_{1},\ldots,e_{4}$ the standard basis of $\C^{4}$,
and set
\begin{equation}
\psi_{\bj,b}(x):=p^{\frac{3n}{2}}\Om(p^{n}\norm{x-\bj}_{p})\,e_{b}\in\HH,\qquad
\bj\in\widetilde{G}_{n}^{3},\ b=1,2,3,4.
\label{psi_J_b}
\end{equation}
Since $\int\Om(p^{n}\norm{x-\bj}_{p})d^{3}x=p^{-3n}$ and the balls are
disjoint, $\{\psi_{\bj,b}\}$ is an orthonormal set in $\HH$. Its closed linear
span is $V_{n}\otimes\C^{4}$, where
\[
V_{n}:=\left\{\phi\in L^{2}(\Qp^{3});\ \phi(x+a)=\phi(x)\text{ for all }a\in p^{n}\Zp^{3}\right\}
\]
is the space of square-integrable functions which are constant on the balls
of radius $p^{-n}$.

\subsection{Coarse-grained position probabilities for a general Hamiltonian}\label{Section_coarse}

Let $\bH$ be a bounded self-adjoint operator on $\HH$, not necessarily
$\bH_{0}$; so $\{e^{-\ii t\bH}\}_{t\in\R}$ is a group of unitary operators on
$\HH$. For $\bi\in\widetilde{G}_{n}^{3}$, let $P_{\bi}$ be the orthogonal
projection on $\HH$ given by multiplication by $\Om(p^{n}\norm{x-\bi}_{p})$;
by (\ref{Partition_A}), $\{P_{\bi}\}_{\bi}$ is a projection-valued measure
(a resolution of the identity: $P_{\bi}P_{\bk}=\delta_{\bi,\bk}P_{\bi}$,
$\sum_{\bi}P_{\bi}=1$ strongly). It describes the measurement of the position
of the particle with resolution $p^{-n}$. Set
\[
\Psi_{\bj,b}(t,x):=e^{-\ii t\bH}\psi_{\bj,b}(x)=\left[\Psi_{\bj,b}^{(1)}(t,x),\ldots,\Psi_{\bj,b}^{(4)}(t,x)\right]^{T},
\]
so that $\norm{\Psi_{\bj,b}(t,\cdot)}=1$ for all $t$, and
\[
\widetilde{\pi}_{\bi,\bj}^{(b)}(t):=\norm{P_{\bi}\Psi_{\bj,b}(t,\cdot)}^{2}
=\sum_{a=1}^{4}\int_{\bi+p^{n}\Zp^{3}}\abs{\Psi_{\bj,b}^{(a)}(t,x)}^{2}d^{3}x.
\]
By the Born rule, $\widetilde{\pi}_{\bi,\bj}^{(b)}(t)$ is the probability
that a measurement of the position, with resolution $p^{-n}$, performed at
time $t$, finds the particle in the ball $\bi+p^{n}\Zp^{3}$, given that at
time $0$ the particle was prepared in the state $\psi_{\bj,b}$, uniformly
localized in the ball $\bj+p^{n}\Zp^{3}$ with internal state $e_{b}$. Using
(\ref{Partition_A}),
\[
\sum_{\bi\in\widetilde{G}_{n}^{3}}\widetilde{\pi}_{\bi,\bj}^{(b)}(t)=\sum_{a=1}^{4}\int_{\Qp^{3}}\abs{\Psi_{\bj,b}^{(a)}(t,x)}^{2}d^{3}x=\norm{\Psi_{\bj,b}(t,\cdot)}^{2}=1,
\]
so, for every $t$ and $b$, $[\widetilde{\pi}_{\bi,\bj}^{(b)}(t)]_{\bi,\bj}$ is a
(column) stochastic matrix indexed by the countable set $\widetilde{G}_{n}^{3}$.

These numbers should be distinguished from the transition probabilities
\[
\pi_{(\bi,a),(\bj,b)}(t):=\abs{(e^{-\ii t\bH}\psi_{\bj,b},\psi_{\bi,a})}^{2}
\]
from the state $\psi_{\bj,b}$ to the state $\psi_{\bi,a}$, which are the ones
used in the theory of CTQWs \cite{Farhi-Gutman}. Since
$\{p^{3n/2}\Om(p^{n}\norm{x-\bi}_{p})e_{a}\}_{a}$ is an orthonormal set in
the range of $P_{\bi}$, Bessel's inequality gives
\begin{equation}
\widetilde{\pi}_{\bi,\bj}^{(b)}(t)\geq\sum_{a=1}^{4}\pi_{(\bi,a),(\bj,b)}(t),
\label{Bessel}
\end{equation}
with equality for all $\bi$ if and only if
$\Psi_{\bj,b}(t,\cdot)\in V_{n}\otimes\C^{4}$. The difference is the
probability carried by the components of the wavefunction which vary inside
the balls of radius $p^{-n}$. Consequently, for a general $\bH$, the
coarse-grained dynamics $t\mapsto\widetilde{\pi}^{(b)}(t)$ is \emph{not} a CTQW
on $\widetilde{G}_{n}^{3}$ in the strict sense: there is, in general, no
Hamiltonian on $\ell^{2}(\widetilde{G}_{n}^{3})\otimes\C^{4}$ generating it. It
is a CTQW exactly when $V_{n}\otimes\C^{4}$ is invariant under $e^{-\ii t\bH}$.
We now show that this is the case for $\bH_{0}$.

\subsection{An exact CTQW for the free $p$-adic Dirac equation}\label{Section_exact_CTQW_I}

\begin{lemma}\label{Lemma_invariance}
Let $J_{1},J_{2},J_{3}$ be arbitrary integrable radial kernels as in Section
\ref{Section_new_type}. Then $V_{n}$ is invariant under $\bJ_{x_{1}}$,
$\bJ_{x_{2}}$, $\bJ_{x_{3}}$, and $V_{n}\otimes\C^{4}$ is invariant under
$\bH_{0}$ and under $e^{-\ii t\bH_{0}}$, $t\in\R$. Moreover, for $l=1,2,3$
and $\phi\in V_{n}$, writing $\phi(\bi)$ for the value of $\phi$ on the ball
$\bi+p^{n}\Zp^{3}$,
\begin{equation}
(\bJ_{x_{l}}\phi)(\bi)=\sum_{k\in\widetilde{G}_{n},\,k\neq i_{l}}w_{l}(i_{l},k)\left\{\phi(\bi^{(l)}(k))-\phi(\bi)\right\},\qquad
w_{l}(i,k):=p^{-n}J_{l}(\abs{i-k}_{p}),
\label{graph_Laplacian}
\end{equation}
where $\bi^{(l)}(k)$ is obtained from $\bi$ by replacing its $l$-th coordinate
by $k$. The series converges absolutely, since
$\sum_{k\neq i}\abs{w_{l}(i,k)}=\int_{\abs{y}_{p}>p^{-n}}\abs{J_{l}(\abs{y}_{p})}dy\leq\norm{J_{l}}_{1}$.
\end{lemma}

\begin{proof}
A function $\phi\in L^{2}(\Qp^{3})$ belongs to $V_{n}$ if and only if
$(\chi_{p}(a\cdot\xi)-1)\widehat{\phi}(\xi)=0$ for all $a\in p^{n}\Zp^{3}$,
i.e., if and only if $\widehat{\phi}$ is supported in the annihilator
$p^{-n}\Zp^{3}$ of $p^{n}\Zp^{3}$. Since $\bJ_{x_{l}}$ is the Fourier
multiplier with symbol $\jj_{l}(\xi_{l})$, it preserves this support
condition; hence $\bJ_{x_{l}}V_{n}\subset V_{n}$, and therefore
$\bH_{0}(V_{n}\otimes\C^{4})\subset V_{n}\otimes\C^{4}$. Since $\bH_{0}$ is
bounded, $e^{-\ii t\bH_{0}}=\sum_{\nu\geq0}(-\ii t\bH_{0})^{\nu}/\nu!$ also preserves
the closed subspace $V_{n}\otimes\C^{4}$. To compute (\ref{graph_Laplacian}),
it is enough to consider one variable. If $x\in i+p^{n}\Zp$ and
$y\in k+p^{n}\Zp$ with $k\neq i$, then $\abs{x-y}_{p}=\abs{i-k}_{p}$ by the
ultrametric property. Hence, for $\phi\in V_{n}$ (one variable),
\begin{align*}
\bJ\phi(x)&=\sum_{k\neq i}\int_{k+p^{n}\Zp}J(\abs{x-y}_{p})\{\phi(k)-\phi(i)\}dy+\int_{i+p^{n}\Zp}J(\abs{x-y}_{p})\{\phi(i)-\phi(i)\}dy\\
&=\sum_{k\neq i}p^{-n}J(\abs{i-k}_{p})\{\phi(k)-\phi(i)\}.
\end{align*}
The bound on $\sum_{k}\abs{w(i,k)}$ follows from
$\{y;\ \abs{y}_{p}>p^{-n}\}=\bigsqcup_{k\neq0}(k+p^{n}\Zp)$.
\end{proof}

Formula (\ref{graph_Laplacian}) says that, on $V_{n}\cong\ell^{2}(\widetilde{G}_{n}^{3})$,
the operator $\bJ_{x_{l}}$ is the weighted graph Laplacian
$\bL^{(l)}:=A^{(l)}-D^{(l)}$ (with the sign convention of
(\ref{graph_Laplacian})) of the graph with vertex set $\widetilde{G}_{n}^{3}$, in
which $\bi$ and $\bk$ are joined by an edge if they differ only in the $l$-th
coordinate, with weight $w_{l}(i_{l},k_{l})$; here $A^{(l)}$ is the weighted
adjacency matrix and $D^{(l)}=d_{l}\mathrm{Id}$, with
$d_{l}:=\int_{\abs{y}_{p}>p^{-n}}J_{l}(\abs{y}_{p})dy$, is the (constant)
weighted degree. The weights depend only on $\abs{i_{l}-k_{l}}_{p}$, that is,
on the level of the $p$-adic tree at which the branches containing $i_{l}$
and $k_{l}$ separate: this is the hierarchical (ultrametric) structure of the
graph. The matrix of $\bL^{(l)}$ in the orthonormal basis
$\{p^{3n/2}\Om(p^{n}\norm{x-\bi}_{p})\}_{\bi}$ is real and symmetric, and
$\norm{\bL^{(l)}}\leq2\norm{J_{l}}_{1}$.

\begin{theorem}\label{Theorem_CTQW_I}
Let $J_{1},J_{2},J_{3}$ be arbitrary integrable radial kernels, and let
$\mathcal{I}_{n}:\ell^{2}(\widetilde{G}_{n}^{3})\otimes\C^{4}\rightarrow V_{n}\otimes\C^{4}\subset\HH$
be the unitary map sending the canonical basis vector $\delta_{\bj}\otimes e_{b}$
to $\psi_{\bj,b}$. Then
\begin{equation}
\mathcal{I}_{n}^{-1}\,\bH_{0}\,\mathcal{I}_{n}=\bH_{n}:=\sum_{l=1}^{3}\bL^{(l)}\otimes\alpha_{l}+m\,\mathrm{Id}\otimes\beta,
\label{H_n}
\end{equation}
a bounded self-adjoint operator on $\ell^{2}(\widetilde{G}_{n}^{3})\otimes\C^{4}$,
and $\mathcal{I}_{n}^{-1}e^{-\ii t\bH_{0}}\mathcal{I}_{n}=e^{-\ii t\bH_{n}}$ for
every $t\in\R$. Consequently, the restriction of the free $p$-adic Dirac
evolution to $V_{n}\otimes\C^{4}$ is the continuous-time quantum walk on the
weighted graph $\widetilde{G}_{n}^{3}$, with internal space $\C^{4}$ at each
vertex, generated by $\bH_{n}$; and, for $\bH=\bH_{0}$, equality holds in
(\ref{Bessel}):
\[
\widetilde{\pi}_{\bi,\bj}^{(b)}(t)=\sum_{a=1}^{4}\abs{\left(e^{-\ii t\bH_{n}}\right)_{(\bi,a),(\bj,b)}}^{2}.
\]
\end{theorem}

\begin{proof}
This follows from Lemma \ref{Lemma_invariance}, since $\bH_{0}$ acts on
$V_{n}\otimes\C^{4}$ as $\sum_{l}\bJ_{x_{l}}\alpha_{l}+m\beta$, and
$\mathcal{I}_{n}$ intertwines $\bJ_{x_{l}}|_{V_{n}}$ with $\bL^{(l)}$. The
equality in (\ref{Bessel}) holds because $e^{-\ii t\bH_{0}}\psi_{\bj,b}\in V_{n}\otimes\C^{4}$.
\end{proof}

\begin{remark}
The CTQW of Theorem \ref{Theorem_CTQW_I} is an exact discretization: no
approximation is involved, because $V_{n}\otimes\C^{4}$ is a reducing
subspace of $\bH_{0}$ (indeed, the orthogonal projection onto it is the
Fourier multiplier $1_{p^{-n}\Zp^{3}}(\xi)$, which commutes with $\bH_{0}$).
The spaces $V_{n}$ increase with $n$, and $\bigcup_{n}V_{n}$ is dense in
$L^{2}(\Qp^{3})$, so the walks for all $n$ together determine the full
evolution $e^{-\ii t\bH_{0}}$.
\end{remark}

%=====================================================================
\section{Continuous-time quantum walks II}\label{Section_CTQW_II}

\subsection{Preliminary results}

In this section, we work with the $p$-adic Dirac Hamiltonian $\bH_{0}$, but
select the kernels $J_{1}$, $J_{2}$, $J_{3}$ such that the condition
\begin{equation}
e^{-\ii t\bH_{0}}\left(L^{2}(\Zp^{3})\otimes\C^{4}\right)\subset L^{2}(\Zp^{3})\otimes\C^{4}
\label{Condition}
\end{equation}
is satisfied for every $t$. Here we identify a function of
$L^{2}(\Zp^{3})$ with a function of $L^{2}(\Qp^{3})$ supported on $\Zp^{3}$.
A sufficient condition for (\ref{Condition}) is the following assumption,
which we impose from now on (see Lemma \ref{Lemma2}).

\begin{assumption}\label{Assumption_A}
For $l=1,2,3$, the kernel $J_{l}$ is supported in the unit ball, i.e.,
$J_{l}(\abs{y}_{p})=0$ for $\abs{y}_{p}>1$, and
$J_{l}(\abs{y}_{p})\in L^{1}(\Zp)$.
\end{assumption}

For a positive integer $n\geq1$, take $G_{n}:=\Zp/p^{n}\Zp$, and use
$i=i_{0}+i_{1}p+\ldots+i_{n-1}p^{n-1}$, with the $i_{k}\in\{0,1,\ldots,p-1\}$,
as a set of representatives of the elements of the quotient group $G_{n}$.
We set
$G_{n}^{3}=(\Zp/p^{n}\Zp)^{3}=\Zp^{3}/p^{n}\Zp^{3}$.
Notice that the group $G_{n}^{3}$ has $p^{3n}$ elements. We use
$\bi=(i_{1},i_{2},i_{3})$ to denote the elements of $G_{n}^{3}$, and fix the
partition
$\Zp^{3}=\bigsqcup_{\bi\in G_{n}^{3}}(\bi+p^{n}\Zp^{3})$.
Taking $\HH_{\Zp}:=L^{2}(\Zp^{3})\otimes\C^{4}$, we have that $\HH_{\Zp}$ is a
closed subspace of $\HH$.

\begin{lemma}\label{Lemma1}
Set, as in (\ref{psi_J_b}),
\begin{equation}
\psi_{\bj,b}(x)=e_{b}p^{\frac{3n}{2}}\Om(p^{n}\norm{x-\bj}_{p}),\qquad b=1,2,3,4,\ \bj\in G_{n}^{3}.
\label{psi_j_J}
\end{equation}
Then $\{\psi_{\bj,b};\ b=1,2,3,4,\ \bj\in G_{n}^{3}\}$ is an orthonormal set
in $\HH$.
\end{lemma}

\begin{proof}
The result follows from
$(e_{b}p^{\frac{3n}{2}}\Om(p^{n}\norm{x-\bj}_{p}),\,e_{a}p^{\frac{3n}{2}}\Om(p^{n}\norm{x-\bi}_{p}))=\delta_{b,a}\delta_{\bi,\bj}$.
\end{proof}

We define the $\C$-vector space $\mathcal{D}_{n}(\Zp^{3})$ as the space
spanned by $\{\Om(p^{n}\norm{x-\bj}_{p})\}_{\bj\in G_{n}^{3}}$, i.e.,
$\mathcal{D}_{n}(\Zp^{3})=V_{n}\cap L^{2}(\Zp^{3})$. Notice that the
cardinality of this set is $p^{3n}$.

\begin{lemma}\label{Lemma2}
Under Assumption \ref{Assumption_A}, condition (\ref{Condition}) holds,
and
\[
e^{-\ii t\bH_{0}}\left(\mathcal{D}_{n}(\Zp^{3})\otimes\C^{4}\right)\subset\mathcal{D}_{n}(\Zp^{3})\otimes\C^{4}\subset L^{2}(\Zp^{3})\otimes\C^{4},\qquad t\in\R.
\]
Moreover, for $\phi\in\mathcal{D}_{n}(\Zp^{3})$, formula
(\ref{graph_Laplacian}) holds with the sum restricted to $k\in G_{n}$,
$k\neq i_{l}$; the corresponding weighted degree is
$d_{l}=\int_{\Zp\smallsetminus p^{n}\Zp}J_{l}(\abs{y}_{p})dy=\sum_{k\in G_{n},k\neq i}w_{l}(i,k)$.
\end{lemma}

\begin{proof}
Let $\phi\in L^{2}(\Qp^{3})$ be supported in $\Zp^{3}$. If $x_{l}\notin\Zp$
and $y\in\Zp$, then $\abs{x_{l}-y}_{p}=\abs{x_{l}}_{p}>1$, so
$J_{l}(\abs{x_{l}-y}_{p})=0$; hence, by (\ref{J_x_l}),
$\bJ_{x_{l}}\phi(x)=0$ whenever $x_{l}\notin\Zp$, and $\bJ_{x_{l}}\phi(x)$
vanishes if some other coordinate $x_{k}\notin\Zp$ because then
$\phi(x)=\phi(x^{(l)}(y))=0$. Thus $\bJ_{x_{l}}$ preserves $L^{2}(\Zp^{3})$, so
$\bH_{0}$ preserves $\HH_{\Zp}$, and so does $e^{-\ii t\bH_{0}}$ (power
series). Since also $V_{n}$ is invariant (Lemma \ref{Lemma_invariance}), the
intersection $\mathcal{D}_{n}(\Zp^{3})\otimes\C^{4}$ is invariant. In
(\ref{graph_Laplacian}), the terms with $k\notin G_{n}$ (i.e., $k\notin\Zp$)
vanish because $J_{l}(\abs{i-k}_{p})=0$ for $\abs{i-k}_{p}>1$. This is
consistent with \cite[Lemma 13.1]{Zuniga-Chacon}: if $J(\abs{x}_{p})$ is an
integrable function supported in the unit ball, then
\begin{align*}
J(\abs{x}_{p})\ast\Om(p^{n}\abs{x-k}_{p})&=p^{-n}\sum_{i\in G_{n}\smallsetminus\{k\}}J(\abs{i-k}_{p})\Om(p^{n}\abs{x-i}_{p})\\
&\quad+\left(\int_{p^{n}\Zp}J(\abs{y}_{p})dy\right)\Om(p^{n}\abs{x-k}_{p}),
\end{align*}
for $k\in G_{n}$.
\end{proof}

\begin{theorem}\label{Theorem_CTQW_II}
Under Assumption \ref{Assumption_A}, let
$\mathcal{I}_{n}:\ell^{2}(G_{n}^{3})\otimes\C^{4}\rightarrow\mathcal{D}_{n}(\Zp^{3})\otimes\C^{4}$
be the unitary map sending $\delta_{\bj}\otimes e_{b}$ to $\psi_{\bj,b}$. Then
$\mathcal{I}_{n}^{-1}\bH_{0}\mathcal{I}_{n}=\bH_{n}$, where $\bH_{n}$ is the
Hermitian $4p^{3n}\times4p^{3n}$ matrix given by (\ref{H_n}), with
$\widetilde{G}_{n}$ replaced by $G_{n}$ and with the weights and degrees of
Lemma \ref{Lemma2}; and
$\mathcal{I}_{n}^{-1}e^{-\ii t\bH_{0}}\mathcal{I}_{n}=e^{-\ii t\bH_{n}}$ for
every $t\in\R$. In other words, the restriction of the free $p$-adic Dirac
evolution to $\mathcal{D}_{n}(\Zp^{3})\otimes\C^{4}$ is the CTQW on the
finite weighted graph $G_{n}^{3}$, with internal space $\C^{4}$, generated by
$\bH_{n}$.
\end{theorem}

\begin{proof}
This follows from Lemmas \ref{Lemma_invariance} and \ref{Lemma2}, exactly as
in the proof of Theorem \ref{Theorem_CTQW_I}.
\end{proof}

\begin{remark}\label{Nota2}
Since the dimension of the $\C$-vector space $\mathcal{D}_{n}(\Zp^{3})\otimes\C^{4}$
is $4p^{3n}$, the restriction of the operator $\bH_{0}$ to
$\mathcal{D}_{n}(\Zp^{3})\otimes\C^{4}$ is given by a Hermitian matrix of
size $4p^{3n}\times4p^{3n}$, namely the matrix $\bH_{n}$ of
(\ref{H_n}), with $\widetilde{G}_{n}$ replaced by $G_{n}$ and with the
weights of Lemma \ref{Lemma2}. Explicitly, in the basis (\ref{psi_j_J}),
\[
\left(\bH_{n}\right)_{(\bi,a),(\bk,b)}=\sum_{l=1}^{3}\left(\bL^{(l)}\right)_{\bi,\bk}(\alpha_{l})_{ab}+m\,\delta_{\bi,\bk}\beta_{ab},
\]
where $(\bL^{(l)})_{\bi,\bk}=w_{l}(i_{l},k_{l})$ if $\bi$ and $\bk$ differ
only in the $l$-th coordinate, $(\bL^{(l)})_{\bi,\bi}=-d_{l}$, and
$(\bL^{(l)})_{\bi,\bk}=0$ otherwise. This restriction is the exact
discretization of $\bH_{0}$ on $G_{n}^{3}$. In addition, the restriction of
$e^{-\ii t\bH_{0}}$ to $\mathcal{D}_{n}(\Zp^{3})\otimes\C^{4}$ is given by the
unitary matrix $e^{-\ii t\bH_{n}}$ of size $4p^{3n}\times4p^{3n}$.
\end{remark}

\subsection{CTQWs II}

Take $\psi_{\bj,b}(x)$ as in (\ref{psi_j_J}), and set
\[
\Psi_{\bj,b}(t,x)=e^{-\ii t\bH_{0}}\psi_{\bj,b}(x)=\left[\Psi_{\bj,b}^{(1)}(t,x),\ldots,\Psi_{\bj,b}^{(4)}(t,x)\right]^{T}\in\mathcal{D}_{n}(\Zp^{3})\otimes\C^{4},\qquad t\in\R.
\]
Then
\begin{equation}
\norm{\Psi_{\bj,b}(t,\cdot)}=1\quad\text{for every }t\in\R. \label{normalization}
\end{equation}
Now, by Lemma \ref{Lemma2}, $\Psi_{\bj,b}(t,\cdot)\in\mathcal{D}_{n}(\Zp^{3})\otimes\C^{4}$,
and by Lemma \ref{Lemma1},
\[
\Psi_{\bj,b}(t,x)=\sum_{a=1}^{4}\sum_{\bi\in G_{n}^{3}}c_{(\bi,a),(\bj,b)}(t)\,e_{a}p^{\frac{3n}{2}}\Om(p^{n}\norm{x-\bi}_{p}),
\]
where
\[
c_{(\bi,a),(\bj,b)}(t)=\left(\Psi_{\bj,b}(t,\cdot),\psi_{\bi,a}\right)=\left(e^{-\ii t\bH_{n}}\right)_{(\bi,a),(\bj,b)}.
\]

\begin{remark}\label{Nota1}
The condition $\Om(p^{n}\norm{x-\bi}_{p})\Om(p^{n}\norm{x-\bj}_{p})=0$, for
$\bi\neq\bj$, which follows from the fact that the balls $\bi+p^{n}\Zp^{3}$,
$\bj+p^{n}\Zp^{3}$ are disjoint, implies that
\[
\abs{\sum_{\bi\in G_{n}^{3}}c_{(\bi,a),(\bj,b)}\,p^{\frac{3n}{2}}\Om(p^{n}\norm{x-\bi}_{p})}^{2}=\sum_{\bi\in G_{n}^{3}}\abs{c_{(\bi,a),(\bj,b)}}^{2}p^{3n}\Om(p^{n}\norm{x-\bi}_{p}).
\]
\end{remark}

We now define the transition probabilities. By the Born rule,
$\abs{c_{(\bi,a),(\bj,b)}(t)}^{2}$ is the probability of finding the system,
at time $t$, in the state $\psi_{\bi,a}$ (vertex $\bi$, internal state
$e_{a}$), given that it was prepared at time $0$ in the state $\psi_{\bj,b}$.
Since $\Psi_{\bj,b}(t,\cdot)\in\mathcal{D}_{n}(\Zp^{3})\otimes\C^{4}$, the
probability of finding the particle at the vertex $\bi$, irrespective of its
internal state, is
\[
\pi_{\bi,\bj}^{(b)}(t):=\sum_{a=1}^{4}\abs{c_{(\bi,a),(\bj,b)}(t)}^{2}=\norm{P_{\bi}\Psi_{\bj,b}(t,\cdot)}^{2}
\]
(the two expressions coincide by Theorem \ref{Theorem_CTQW_I} and Remark
\ref{Nota1}). If the internal state at time $0$ is not known, i.e., if the
initial state is the mixed state
$\rho_{\bj}:=\frac{1}{4}\sum_{b=1}^{4}\psi_{\bj,b}\otimes\overline{\psi_{\bj,b}}$
(the density operator $\frac{1}{4}\sum_{b}\vert\psi_{\bj,b}\rangle\langle\psi_{\bj,b}\vert$),
then the probability of finding the particle at $\bi$ at time $t$ is
\[
\pi_{\bi,\bj}(t):=\operatorname{Tr}\left(P_{\bi}\,e^{-\ii t\bH_{0}}\rho_{\bj}e^{\ii t\bH_{0}}\right)=\frac{1}{4}\sum_{b=1}^{4}\pi_{\bi,\bj}^{(b)}(t)=\frac{1}{4}\sum_{a,b=1}^{4}\abs{c_{(\bi,a),(\bj,b)}(t)}^{2}.
\]
Then, from (\ref{normalization}) and Remark \ref{Nota1}, we obtain that
\[
1=\sum_{a=1}^{4}\norm{\sum_{\bi\in G_{n}^{3}}c_{(\bi,a),(\bj,b)}p^{\frac{3n}{2}}\Om(p^{n}\norm{x-\bi}_{p})}_{2}^{2}
=\sum_{a=1}^{4}\sum_{\bi\in G_{n}^{3}}\abs{c_{(\bi,a),(\bj,b)}}^{2}=\sum_{\bi\in G_{n}^{3}}\pi_{\bi,\bj}^{(b)}(t).
\]
We have thus established the following result.

\begin{theorem}\label{Theorem_stochastic}
Under Assumption \ref{Assumption_A}, for every $t\in\R$:
\begin{enumerate}[label=(\roman*)]
\item for every $b\in\{1,2,3,4\}$, the matrix
$[\pi_{\bi,\bj}^{(b)}(t)]_{\bi,\bj\in G_{n}^{3}}$ is stochastic, i.e.,
$\pi_{\bi,\bj}^{(b)}(t)\in[0,1]$ and, for any $\bj\in G_{n}^{3}$,
$\sum_{\bi\in G_{n}^{3}}\pi_{\bi,\bj}^{(b)}(t)=1$;
\item the matrix $[\pi_{\bi,\bj}(t)]_{\bi,\bj\in G_{n}^{3}}$ is doubly
stochastic, i.e., $\sum_{\bi}\pi_{\bi,\bj}(t)=1$ for every $\bj$, and
$\sum_{\bj}\pi_{\bi,\bj}(t)=1$ for every $\bi$.
\end{enumerate}
Consequently, for each fixed $t$, $\pi_{\bi,\bj}(t)$ is the one-step
transition kernel of a random walk on the state space $G_{n}^{3}$, and the
uniform distribution on $G_{n}^{3}$ is invariant under it.
\end{theorem}

\begin{proof}
Part (i) was established above. For (ii), the column sums equal $1$ by
(i). For the row sums, let $U=e^{-\ii t\bH_{n}}$, a unitary matrix with
entries $U_{(\bi,a),(\bj,b)}=c_{(\bi,a),(\bj,b)}(t)$. Each row of a unitary
matrix has norm one, so
$\sum_{\bj}\pi_{\bi,\bj}(t)=\frac{1}{4}\sum_{a=1}^{4}\sum_{\bj,b}\abs{U_{(\bi,a),(\bj,b)}}^{2}=\frac{1}{4}\cdot4=1$.
\end{proof}

The family $\{[\pi_{\bi,\bj}(t)]\}_{t\geq0}$ is \emph{not}, however, the family
of transition matrices of a continuous-time Markov chain. This is a general
feature of CTQWs, see \cite{Farhi-Gutman,Childs-et-al}, and it is the
precise sense in which a CTQW differs from a classical random walk:

\begin{proposition}\label{Prop_non_Markov}
Assume that $\bH_{n}$ has a non-zero matrix entry between two different
vertices, i.e., $(\bH_{n})_{(\bi,a),(\bk,b)}\neq0$ for some $\bi\neq\bk$ and
some $a,b$ (this is the case as soon as some $w_{l}(i,k)\neq0$). Then there
is no matrix-valued function $\Pi(t)$ satisfying $\Pi(t+s)=\Pi(t)\Pi(s)$ for
all $t,s\geq0$ and $\Pi(t)=[\pi_{\bi,\bj}(t)]$ for all $t\geq0$; the same holds
for $[\pi_{\bi,\bj}^{(b)}(t)]$, $b=1,2,3,4$.
\end{proposition}

\begin{proof}
Write $\pi(t)=[\pi_{\bi,\bj}(t)]$. Since $e^{-\ii t\bH_{n}}=\mathrm{Id}-\ii
t\bH_{n}+O(t^{2})$, for $\bi\neq\bj$ we have
$c_{(\bi,a),(\bj,b)}(t)=-\ii t(\bH_{n})_{(\bi,a),(\bj,b)}+O(t^{2})$, so
$\pi_{\bi,\bj}(t)=\frac{t^{2}}{4}\sum_{a,b}\abs{(\bH_{n})_{(\bi,a),(\bj,b)}}^{2}+O(t^{3})$,
and $\pi(t)=\mathrm{Id}+O(t^{2})$ (the diagonal entries are determined by the
column sums). Suppose $\pi(t+s)=\pi(t)\pi(s)$ for $t,s\geq0$. Then $\pi$ is a
continuous semigroup of finite matrices with $\pi(0)=\mathrm{Id}$, so
$\pi(t)=e^{tQ}$ with $Q=\pi^{\prime}(0)=0$; hence $\pi(t)=\mathrm{Id}$ for all
$t\geq0$. But, for the pair $\bi\neq\bk$ in the hypothesis,
$\pi_{\bi,\bk}(t)\geq\frac{t^{2}}{4}\abs{(\bH_{n})_{(\bi,a),(\bk,b)}}^{2}+O(t^{3})>0$
for small $t>0$, a contradiction. The same argument applies to
$\pi^{(b)}(t)$, with $b$ as in the hypothesis.
\end{proof}

Thus, tracing out the internal (spinor) degrees of freedom of the $p$-adic
Dirac CTQW produces, at every fixed time $t$, a stochastic (and, after
averaging over the initial internal state, doubly stochastic) matrix on the
finite graph $G_{n}^{3}$; the time evolution of these matrices is quantum,
not Markovian.

\begin{remark}
By Parseval's identity, the amplitudes can be computed in the Fourier
representation:
$c_{(\bi,a),(\bj,b)}(t)=\langle\widehat{\Psi}_{\bj,b}^{(a)}(t,\xi),p^{\frac{-3n}{2}}\chi_{p}(\bi\cdot\xi)\Om(\norm{p^{n}\xi}_{p})\rangle$,
since the Fourier transform of $p^{\frac{3n}{2}}\Om(p^{n}\norm{x-\bi}_{p})$
is $p^{\frac{-3n}{2}}\chi_{p}(\bi\cdot\xi)\Om(\norm{p^{n}\xi}_{p})$.
\end{remark}

\subsection{Computation of $\Psi_{\bj,b}(t,x)$}

Consider the Cauchy problem
\begin{equation}
\left\{
\begin{array}{ll}
\ii\frac{\partial}{\partial t}\Psi_{\bj,b}(t,x)=\F_{\xi\rightarrow x}^{-1}\left\{h(\xi)\F_{x\rightarrow\xi}\left(\Psi_{\bj,b}(t,x)\right)\right\}, & x\in\Qp^{3},\ t\in\R,\\[4pt]
\Psi_{\bj,b}(0,x)=\psi_{\bj,b}(x)\in\HH_{\Zp}. &
\end{array}
\right.
\label{Cauchy_Problem_1}
\end{equation}
By passing to the Fourier transform in (\ref{Cauchy_Problem_1}),
\[
\left\{
\begin{array}{ll}
\ii\frac{\partial}{\partial t}\widehat{\Psi_{\bj,b}}(t,\xi)=h(\xi)\widehat{\Psi_{\bj,b}}(t,\xi), & \xi\in\Qp^{3},\ t\in\R,\\[4pt]
\widehat{\Psi_{\bj,b}}(0,\xi)=\widehat{\psi_{\bj,b}}(\xi), &
\end{array}
\right.
\]
and using $h(\xi)=u(\xi)^{-1}\beta\lambda(\xi)u(\xi)$, cf.
(\ref{diagonalization}), we have
\[
\Psi_{\bj,b}(t,x)=\F_{\xi\rightarrow x}^{-1}\left\{u(\xi)^{-1}e^{-\ii t\beta\lambda(\xi)}u(\xi)\widehat{\psi_{\bj,b}}(\xi)\right\}\in\mathcal{D}_{n}(\Zp^{3})\otimes\C^{4},\qquad t\in\R.
\]
Since $h(\xi)^{2}=\lambda(\xi)^{2}\Id$, the matrix exponential can be
written without diagonalizing:
\[
u(\xi)^{-1}e^{-\ii t\beta\lambda(\xi)}u(\xi)=e^{-\ii th(\xi)}=\cos(t\lambda(\xi))\Id-\ii\,\frac{\sin(t\lambda(\xi))}{\lambda(\xi)}h(\xi),
\]
with the convention $\sin(t\lambda)/\lambda=t$ if $\lambda=0$. Using
$\widehat{\psi_{\bj,b}}(\xi)=p^{-3n/2}\chi_{p}(\bj\cdot\xi)\Om(\norm{p^{n}\xi}_{p})e_{b}$,
we obtain the explicit formula for the amplitudes
\[
c_{(\bi,a),(\bj,b)}(t)=p^{-3n}\int_{\norm{\xi}_{p}\leq p^{n}}\chi_{p}((\bj-\bi)\cdot\xi)\left\{\cos(t\lambda(\xi))\delta_{ab}-\ii\frac{\sin(t\lambda(\xi))}{\lambda(\xi)}h_{ab}(\xi)\right\}d^{3}\xi.
\]
Under Assumption \ref{Assumption_A}, $\bjj(\xi)$ vanishes on $\Zp^{3}$ and
takes finitely many values on $\{\norm{\xi}_{p}\leq p^{n}\}$ (it depends only
on $\abs{\xi_{l}}_{p}\in\{p^{0},\ldots,p^{n}\}$), so the integral is a finite
sum of integrals of additive characters over products of $p$-adic spheres,
which can be evaluated in closed form.

%=====================================================================
\section{Towards quantum networks based on the Dirac equation}\label{Section_Quantum_Computers}

The constructions of Sections \ref{Section_CTQW_I} and \ref{Section_CTQW_II}
were carried out for their own sake, as exact discretizations of the
free $p$-adic Dirac equation into CTQWs. We now argue that these results
suggest the feasibility of quantum networks whose elementary dynamics are
Dirac-, rather than Schr\"{o}dinger-, type, and discuss how they relate to the
broader state of the art reviewed in Section \ref{Section_Related_Work}.

\subsection{From a single walk to a quantum network}

Theorem \ref{Theorem_stochastic} shows that the matrix
$[\pi_{\bi,\bj}(t)]_{\bi,\bj\in G_{n}^{3}}$, obtained by summing the Born-rule
weights of the four spinor components at each node and averaging over
the initial internal state, is a doubly stochastic matrix for every
$t$. In other words, tracing out the internal degrees of freedom of the
$p$-adic Dirac CTQW produces, at every fixed time, the transition kernel
of a random walk on the finite graph $G_{n}^{3}$, while the time
dependence remains genuinely quantum (Proposition \ref{Prop_non_Markov}).
This is the kind of object used as the substrate for the quantum neural
network introduced in \cite{Zuniga-QNN}. In that paper, the authors present a
new class of quantum neural networks (QNNs) whose states are solutions of
$p$-adic Schr\"{o}dinger equations with a non-local potential that controls
the interaction between the neurons. These equations are obtained as Wick
rotations of the state equations of $p$-adic cellular neural networks (CNNs),
\cite{Zambrano-Zuniga-1,Zambrano-Zuniga-2,Zuniga-et-al}. $p$-Adic CNNs arise
as continuous limits of large discrete hierarchical neural networks (NNs).
They are bio-inspired by the Wilson--Cowan model, which describes the
macroscopic dynamics of large neuronal populations. Here, we propose a similar
construction using the $p$-adic Dirac equations introduced here.

First, we introduce a kernel matrix
\[
W(x,y)=\left[W_{ab}(x,y)\right]_{a,b=1}^{4},\qquad x,y\in\Qp^{3},
\]
where $W_{ab}(x,y)\in\C$. The entry $W_{ab}(x,y)$ describes the strength of
the connection between the $b$-th spinor component at position $y$ and the
$a$-th spinor component at position $x$. We fix an activation function
$\phi:\R\rightarrow\R$; for instance $\phi(s)=\tanh(s)$ or
$\phi(s)=\frac{1}{2}\abs{s+1}-\frac{1}{2}\abs{s-1}$. We extend $\phi$ to the
complex numbers by taking $\phi(a+\ii b)=\phi(a)+\ii\phi(b)$, and to $\C^{4}$
componentwise:
\[
\phi\left(\Psi(t,y)\right)=\phi\left(\begin{bmatrix}\Psi^{(1)}(t,y)\\\vdots\\\Psi^{(4)}(t,y)\end{bmatrix}\right)
=\begin{bmatrix}\phi(\operatorname{Re}\Psi^{(1)}(t,y))+\ii\phi(\operatorname{Im}\Psi^{(1)}(t,y))\\\vdots\\\phi(\operatorname{Re}\Psi^{(4)}(t,y))+\ii\phi(\operatorname{Im}\Psi^{(4)}(t,y))\end{bmatrix}\in\C^{4},
\]
and we fix a bias
$z(t,x)=[z^{(1)}(t,x),\ldots,z^{(4)}(t,x)]^{T}$. With this notation, we
introduce the Dirac analogue of the QNN introduced in \cite{Zuniga-QNN}:
\begin{equation}
\ii\frac{\partial}{\partial t}\Psi(t,x)=\bH_{0}\Psi(t,x)+\int_{\Qp^{3}}W(x,y)\,\phi\left(\Psi(t,y)\right)d^{3}y+z(t,x),\qquad t\geq0,\ x\in\Qp^{3},
\label{Dirac_QNN}
\end{equation}
where $W(x,y)\phi(\Psi(t,y))$ denotes the product of a $4\times4$ matrix and a
vector of $\C^{4}$. We make the following assumptions:
\begin{enumerate}[label=(H\arabic*)]
\item $W_{ab}\in L^{2}(\Qp^{3}\times\Qp^{3})$ for $a,b=1,\ldots,4$;
\item $\phi:\R\rightarrow\R$ is Lipschitz continuous, with Lipschitz constant
$L_{\phi}$, and $\phi(0)=0$;
\item $z\in C([0,\infty),\HH)$.
\end{enumerate}
Both activation functions mentioned above satisfy (H2) with $L_{\phi}=1$. The
condition $\phi(0)=0$ cannot be dropped: since $\Qp^{3}$ has infinite
measure, a non-zero constant does not belong to $L^{2}(\Qp^{3})$.

\begin{lemma}\label{Lemma_Nemytskii}
Assume (H1)--(H2). Then:
\begin{enumerate}[label=(\roman*)]
\item the integral operator $(\mathcal{W}\Phi)(x):=\int_{\Qp^{3}}W(x,y)\Phi(y)d^{3}y$
is a bounded (Hilbert--Schmidt) operator on $\HH$, with
$\norm{\mathcal{W}}\leq\norm{W}_{2}:=(\sum_{a,b}\norm{W_{ab}}_{L^{2}(\Qp^{3}\times\Qp^{3})}^{2})^{1/2}$;
\item the Nemytskii (superposition) operator $\Phi(\Psi)(y):=\phi(\Psi(y))$
maps $\HH$ into $\HH$, and
$\norm{\Phi(\Psi)-\Phi(\Psi^{\prime})}\leq L_{\phi}\norm{\Psi-\Psi^{\prime}}$
for all $\Psi,\Psi^{\prime}\in\HH$.
\end{enumerate}
\end{lemma}

\begin{proof}
(i) By the Cauchy--Schwarz inequality,
$\abs{(\mathcal{W}\Phi)_{a}(x)}^{2}\leq\left(\sum_{b}\int\abs{W_{ab}(x,y)}^{2}d^{3}y\right)\norm{\Phi}^{2}$;
integrate in $x$ and sum over $a$. (ii) For $a,b,a^{\prime},b^{\prime}\in\R$,
$\abs{\phi(a)+\ii\phi(b)-\phi(a^{\prime})-\ii\phi(b^{\prime})}^{2}\leq L_{\phi}^{2}(\abs{a-a^{\prime}}^{2}+\abs{b-b^{\prime}}^{2})$.
Applying this componentwise and pointwise, and integrating, gives the
Lipschitz estimate; taking $\Psi^{\prime}=0$ and using $\phi(0)=0$ gives
$\norm{\Phi(\Psi)}\leq L_{\phi}\norm{\Psi}$, so $\Phi(\HH)\subset\HH$. See
\cite{Appell-Zabrejko} for the general theory of superposition operators.
\end{proof}

\begin{proposition}\label{Prop_well_posed}
Assume (H1)--(H3). Then, for every $\Psi_{0}\in\HH$, equation
(\ref{Dirac_QNN}) has a unique mild solution $\Psi\in C([0,\infty),\HH)$
with $\Psi(0)=\Psi_{0}$, i.e., a unique continuous solution of
\[
\Psi(t)=e^{-\ii t\bH_{0}}\Psi_{0}-\ii\int_{0}^{t}e^{-\ii(t-s)\bH_{0}}\left\{\mathcal{W}\Phi(\Psi(s))+z(s)\right\}ds.
\]
Moreover, $\Psi\in C^{1}([0,\infty),\HH)$ and (\ref{Dirac_QNN}) holds in $\HH$
for every $t\geq0$.
\end{proposition}

\begin{proof}
By Lemma \ref{Lemma_Nemytskii}, the map
$F(t,\Psi):=-\ii\{\mathcal{W}\Phi(\Psi)+z(t)\}$ is continuous in $t$ and
globally Lipschitz in $\Psi$, with constant $\norm{W}_{2}L_{\phi}$, uniformly
in $t$. Since $e^{-\ii t\bH_{0}}$ is a uniformly continuous unitary group,
the standard Picard iteration (equivalently, the Banach fixed-point theorem
in $C([0,T],\HH)$ with an exponentially weighted norm) gives a unique global
mild solution, see \cite[Theorem 6.1.2]{Pazy}; since $\bH_{0}$ is bounded and
$F$ is continuous, the mild solution is a classical $C^{1}$ solution.
\end{proof}

In a discretization of (\ref{Dirac_QNN}) on $G_{n}^{3}$, i.e., restricting to
$\mathcal{D}_{n}(\Zp^{3})\otimes\C^{4}$ and replacing $W$ by a
$4p^{3n}\times4p^{3n}$ matrix, (\ref{Dirac_QNN}) becomes a system of
$4p^{3n}$ ordinary differential equations with globally Lipschitz
right-hand side, which again has unique global solutions.

\begin{remark}\label{Remark_non_unitary}
For $W\neq0$ or $z\neq0$, the evolution (\ref{Dirac_QNN}) is nonlinear and,
in general, it does not preserve the norm: for a solution $\Psi$, using that
$(\bH_{0}\Psi,\Psi)$ is real,
\[
\frac{d}{dt}\norm{\Psi(t)}^{2}=2\operatorname{Im}\left(\mathcal{W}\Phi(\Psi(t))+z(t),\Psi(t)\right),
\]
which is in general non-zero. Hence, for $W\neq0$, the network dynamics is
not unitary, and $\norm{\Psi(t)}^{2}$ cannot, in general, be interpreted as a
total probability without renormalization. We stress that (\ref{Dirac_QNN})
is a nonlinear equation for a vector state, and therefore it is \emph{not} of
Lindblad (Gorini--Kossakowski--Sudarshan--Lindblad) type, since GKSL
equations are linear equations for density operators
\cite{GKS,Lindblad,Breuer-Petruccione}. Whether the dynamics
(\ref{Dirac_QNN}) can be obtained as the reduced dynamics of a larger, closed
or open, quantum system---as is suggested by the interpretation of
$W\phi(\Psi)$ as an effective interaction with an environment---is an open
problem.
\end{remark}

In such a discretization, the vertices $\bj\in G_{n}^{3}$ play the role of neurons or network nodes, arranged, as in
\cite{Zuniga-QNN}, in the hierarchical, ultrametric tree structure inherited
from $\Zp^{3}$, while the four-dimensional fiber $\C^{4}$ attached to each
node carries four-component spinor degrees of freedom that are absent from the
scalar $p$-adic Schr\"{o}dinger networks studied previously. A Dirac-based
quantum network built along these lines would differ from the QNN family
studied in \cite{Zuniga-QNN} in a structurally new way: the free part of
the dynamics has the charge-conjugation symmetry of Section \ref{Section_4},
which on the graph reads $(\mathrm{Id}\otimes\alpha_{2})\overline{\bH_{n}}(\mathrm{Id}\otimes\alpha_{2})^{-1}=-\bH_{n}$
(because the matrices $\bL^{(l)}$ are real) and relates pairs of internal
channels at each node, and the topologically protected zero modes of the
$p$-adic Jackiw--Rebbi model, \cite{Zuniga-JR}, are natural candidates for
noise-resilient network nodes. Adding the non-local weight kernel $W$
and the nonlinear activation $\phi$ turns the free, unitary Dirac CTQW
($W=0$) into a nonlinear, non-unitary network dynamics ($W\neq0$), whose
physical interpretation is discussed in Remark \ref{Remark_non_unitary}.

\subsection{Relation to the broader state of the art}

As discussed in Section \ref{Section_Related_Work}, most of the literature
connecting quantum walks to the Dirac equation is built on discrete-time,
coined walks
\cite{Bialynicki-Birula,Meyer,Succi-Benzi,Arrighi-Nesme-Forets,Strauch1,Strauch2,Kurzynski,Chandrashekar},
with the Dirac equation appearing as a continuum \emph{limit}. Continuous-time
dynamics with spinor-like structure on graphs have also been constructed,
notably the Dirac-type CTQWs on periodic lattices of Childs and Goldstone
\cite{Childs-Goldstone}, the topological Dirac equation on networks of
Bianconi \cite{Bianconi}, CTQWs on bipartite graphs \cite{Todtli}, and
honeycomb tight-binding models \cite{Wallace,Castro-Neto}. The construction
developed in this paper differs from these in three respects: (a) it is
formulated on hierarchical, ultrametric graphs, the edge weights being
determined by the $p$-adic tree; (b) the walk is not an approximation of the
Dirac equation, but the exact restriction of the $p$-adic Dirac evolution to
an invariant subspace (Theorems \ref{Theorem_CTQW_I} and
\ref{Theorem_CTQW_II}); and (c) it is set up as the substrate of a quantum
network. To the best of our knowledge, it is the first CTQW with these
properties.

%=====================================================================
\section{Discussion and conclusions}\label{Section_Discussion}

\subsection{Summary}

We introduced a new $p$-adic version of the Dirac equation, formulated in the
standard axiomatic (Dirac--von Neumann) framework of quantum mechanics, in
which the usual momentum operators are replaced by bounded non-local operators
with real radial kernels. Working in momentum space, we constructed its
plane-wave solutions (Section \ref{Section_Plane_Waves}), diagonalized the
free Hamiltonian, and determined its spectrum (Section \ref{Section_3}),
which is pure point, with eigenvalues of infinite multiplicity; in the
massless case we described the zero-energy subspace. We identified the
charge-conjugation symmetry of the model (Section \ref{Section_4}): it is
implemented by $\mathcal{C}\Psi=\alpha_{2}\overline{\Psi}$, satisfies
$\mathcal{C}^{2}=-1$, exchanges positive- and negative-energy states of the
free Hamiltonian, and, in an external field, reverses both the charge and the
vector potential.

Building on this, we discretized the free equation in two ways, both giving
CTQWs rather than discrete-time, coined walks. The first discretization is
set on a countable graph (Section \ref{Section_CTQW_I}); the second, more
explicit one is set on a finite, tree-like graph (Section
\ref{Section_CTQW_II}). In both cases the discretization is exact, and
the walk Hamiltonian is a Dirac-type combination of weighted graph Laplacians
with ultrametric weights. For the second construction we proved that the
transition probabilities, once the four internal components of the
wavefunction are added together, form a stochastic matrix at every moment in
time, which becomes doubly stochastic after averaging over the initial
internal state; as for every non-trivial CTQW, the time evolution of these
matrices is not Markovian.

We then used these results to argue, in Section
\ref{Section_Quantum_Computers}, that the same construction can serve as the
foundation for a quantum network, extending earlier work on $p$-adic quantum
neural networks \cite{Zuniga-QNN} from scalar (Schr\"{o}dinger-type) internal
states to four-component spinor states, and we proved the global
well-posedness of the network equation under Lipschitz assumptions.

\subsection{Open problems}

The results above raise more questions than they close. We list the ones we
regard as the most pressing.

\begin{enumerate}
\item[(1)] \textbf{Quantitative continuum limit.} Although each walk
of Sections \ref{Section_CTQW_I} and \ref{Section_CTQW_II} is an exact
restriction of the $p$-adic Dirac evolution, the initial states that are not
constant on balls of radius $p^{-n}$ are only approximated by the walk at
level $n$. We have not established a quantitative estimate, analogous to
the convergence rates proved by Arrighi, Nesme, and Forets for the
higher-dimensional Archimedean DTQW--Dirac correspondence
\cite{Arrighi-Nesme-Forets}, for the convergence of
$P_{V_{n}}e^{-\ii t\bH_{0}}P_{V_{n}}\psi$ to $e^{-\ii t\bH_{0}}\psi$ as
$n\rightarrow\infty$ in terms of the regularity of $\psi$. Since
$\Zp^{3}$ is compact and ultrametric, one may expect a rate governed by
$p^{-n}$; making this precise is open.

\item[(2)] \textbf{Interactions and gauge fields.} The Dirac equation studied
here is free. Section \ref{Section_4} shows that, with the standard
minimal coupling, charge conjugation also reverses the vector potential.
Coupling $\bH_{0}$ to a genuinely $p$-adic electromagnetic (or more general
gauge) potential, in the spirit of the gauge-covariant discrete-time walks of
Arnault and Debbasch \cite{Arnault-Debbasch} and of
\cite{DiMolfetta-Brachet-Debbasch} in the Archimedean setting, has not been
carried out; nor is it known whether the resulting equation would still
discretize into a CTQW with a stochastic-matrix structure analogous to the
one established in Section \ref{Section_CTQW_II}.

\item[(3)] \textbf{Momentum-type operators with odd symbols.} The
operators $\bJ_{x_{l}}$ have even (radial) symbols, which is the origin of the
differences with the standard charge conjugation (Remark
\ref{Remark_charge_conj}). Operators of the form $\ii K\ast$, with $K$ real
and odd ($K(-y)=-K(y)$), are self-adjoint and anticommute with complex
conjugation, like $-\ii\partial_{x}$; such kernels exist on $\Qp$ (e.g.,
$K(y)=f(\abs{y}_{p})\sin(2\pi\{y\}_{p})$), but they are not radial. Whether a $p$-adic Dirac equation built from
such operators has good spectral and discretization properties is open.

\item[(4)] \textbf{Curved $p$-adic space-time.} Di Molfetta, Brachet, and
Debbasch showed that a family of discrete-time walks converges to massless
Dirac fermions on a curved background metric
\cite{DiMolfetta-Brachet-Debbasch}. Whether a meaningful notion of curvature
exists for $\Qp^{3}$, and whether the CTQW construction of this paper can be
adapted to such a background, is entirely open.

\item[(5)] \textbf{Nonlinear network dynamics.} Proposition
\ref{Prop_well_posed} gives global well-posedness of the network equation
(\ref{Dirac_QNN}) under Lipschitz assumptions. Its long-time behaviour, the
role of the charge-conjugation symmetry when $W\neq0$, and the question of
whether it can be derived as the reduced dynamics of an open quantum system
(Remark \ref{Remark_non_unitary}) remain open.

\item[(6)] \textbf{Entanglement structure.} We have not characterized the
entanglement generated by the free Dirac CTQW across the hierarchical levels
of the $p$-adic tree $G_{n}^{3}$, nor the entanglement between the four spinor
channels at a fixed node. Understanding this structure is a prerequisite for
assessing whether the construction offers any genuine advantage, as a quantum
network, beyond that already available for scalar CTQWs.

\item[(7)] \textbf{Topological protection on the walk.} The $p$-adic
Jackiw--Rebbi model of \cite{Zuniga-JR} produces topologically protected,
localized zero modes for a $p$-adic Dirac Hamiltonian. We have not
constructed the CTQW that discretizes that model, nor determined whether its
zero modes survive discretization on $G_{n}^{3}$ as noise-resilient network
nodes, as suggested informally in Section \ref{Section_Quantum_Computers}.
\end{enumerate}

We regard item (5) as the most immediate next step, since it extends the
machinery already developed in \cite{Zuniga-QNN} for the scalar case to the
four-spinor setting of the present paper; items (2), (3), (4), and (7), by
contrast, appear to require genuinely new ideas, specific to the interplay
between $p$-adic analysis and the structure of the Dirac equation, and we
leave them for future work.

%=====================================================================
\section{Appendix: Basic facts on $p$-adic analysis}\label{Appendix}

In this section, we fix the notation and collect some basic results on
$p$-adic analysis that we use throughout the article. For a detailed
exposition on $p$-adic analysis, the reader may consult
\cite{V-V-Z,Alberioetal,Zuniga-Textbook,Taibleson}.

\subsection{The field of $p$-adic numbers}

Along this article $p$ denotes a prime number. The field of $p$-adic numbers
$\Qp$ is defined as the completion of the field of rational numbers
$\mathbb{Q}$ with respect to the $p$-adic norm $\abs{\cdot}_{p}$, which is
defined as
\[
\abs{x}_{p}=\begin{cases}0&\text{if }x=0,\\ p^{-\gamma}&\text{if }x=p^{\gamma}\dfrac{a}{b},\end{cases}
\]
where $a$ and $b$ are integers coprime with $p$. The integer
$\gamma=ord_{p}(x):=ord(x)$, with $ord(0):=+\infty$, is called the $p$-adic
order of $x$. We extend the $p$-adic norm to $\Qp^{N}$ by taking
$\norm{x}_{p}:=\max_{1\leq i\leq N}\abs{x_{i}}_{p}$ for
$x=(x_{1},\dots,x_{N})\in\Qp^{N}$. By defining
$ord(x)=\min_{1\leq i\leq N}\{ord(x_{i})\}$, we have $\norm{x}_{p}=p^{-ord(x)}$.
The metric space $(\Qp^{N},\norm{\cdot}_{p})$ is a complete ultrametric space.
As a topological space $\Qp$ is homeomorphic to a Cantor-like subset of the
real line, see, e.g., \cite{V-V-Z,Alberioetal}.

Any $p$-adic number $x\neq0$ has a unique expansion of the form
$x=p^{ord(x)}\sum_{j=0}^{\infty}x_{j}p^{j}$,
where $x_{j}\in\{0,1,2,\dots,p-1\}$ and $x_{0}\neq0$. In addition, any
$x\in\Qp^{N}\smallsetminus\{0\}$ can be represented uniquely as
$x=p^{ord(x)}v$, where $\norm{v}_{p}=1$.

\subsection{Topology of $\Qp^{N}$}

For $r\in\mathbb{Z}$, denote by $B_{r}^{N}(a)=\{x\in\Qp^{N};\norm{x-a}_{p}\leq p^{r}\}$
the ball of radius $p^{r}$ with center at $a=(a_{1},\dots,a_{N})\in\Qp^{N}$,
and take $B_{r}^{N}(0):=B_{r}^{N}$. Note that
$B_{r}^{N}(a)=B_{r}(a_{1})\times\cdots\times B_{r}(a_{N})$, where
$B_{r}(a_{i}):=\{x\in\Qp;\abs{x_{i}-a_{i}}_{p}\leq p^{r}\}$ is the
one-dimensional ball of radius $p^{r}$ with center at $a_{i}\in\Qp$. The ball
$B_{0}^{N}$ equals the product of $N$ copies of $B_{0}=\Zp$, the ring of
$p$-adic integers. We also denote by
$S_{r}^{N}(a)=\{x\in\Qp^{N};\norm{x-a}_{p}=p^{r}\}$ the sphere of radius
$p^{r}$ with center at $a$, and take $S_{r}^{N}(0):=S_{r}^{N}$. In addition,
two balls in $\Qp^{N}$ are either disjoint or one is contained in the other.

As a topological space $(\Qp^{N},\norm{\cdot}_{p})$ is totally disconnected,
i.e., the only connected subsets of $\Qp^{N}$ are the empty set and the
points. A subset of $\Qp^{N}$ is compact if and only if it is closed and
bounded in $\Qp^{N}$, see, e.g., \cite[Section 1.3]{V-V-Z}, or
\cite[Section 1.8]{Alberioetal}. The balls and spheres are compact subsets.
Thus $(\Qp^{N},\norm{\cdot}_{p})$ is a locally compact topological space.

\subsection{The Haar measure}

Since $(\Qp^{N},+)$ is a locally compact topological group, there exists a
Haar measure $d^{N}x$, which is invariant under translations, i.e.,
$d^{N}(x+a)=d^{N}x$, \cite{Halmos}. If we normalize this measure by the
condition $\int_{\Zp^{N}}d^{N}x=1$, then $d^{N}x$ is unique.

\begin{notation}
We will use $\Om(p^{-r}\norm{x-a}_{p})$ to denote the characteristic function
of the ball $B_{r}^{N}(a)=a+p^{-r}\Zp^{N}$, where
$\Zp^{N}=\{x\in\Qp^{N};\norm{x}_{p}\leq1\}$ is the $N$-dimensional unit ball.
For more general sets, we will use the notation $1_{A}$ for the
characteristic function of the set $A$.
\end{notation}

\subsection{The Bruhat--Schwartz space}

A complex-valued function $\varphi$ defined on $\Qp^{N}$ is called locally
constant if for any $x\in\Qp^{N}$ there exists an integer $l(x)\in\mathbb{Z}$
such that
\begin{equation}
\varphi(x+x^{\prime})=\varphi(x)\text{ for any }x^{\prime}\in B_{l(x)}^{N}.
\label{local_constancy}
\end{equation}
A function $\varphi:\Qp^{N}\rightarrow\C$ is called a Bruhat--Schwartz
function (or a test function) if it is locally constant with compact support.
Any test function can be represented as a linear combination, with complex
coefficients, of characteristic functions of balls. The $\C$-vector space of
Bruhat--Schwartz functions is denoted by $\mathcal{D}(\Qp^{N})$. For
$\varphi\in\mathcal{D}(\Qp^{N})$, the largest number $l=l(\varphi)$
satisfying (\ref{local_constancy}) is called the exponent of local constancy
(or the parameter of constancy) of $\varphi$.

\subsection{$L^{\rho}$ spaces}

Given $\rho\in[1,\infty)$, we denote by $L^{\rho}(\Qp^{N}):=L^{\rho}(\Qp^{N},d^{N}x)$
the $\C$-vector space of all the complex valued functions $g$ satisfying
$\norm{g}_{\rho}=\left(\int_{\Qp^{N}}\abs{g(x)}^{\rho}d^{N}x\right)^{1/\rho}<\infty$,
where $d^{N}x$ is the normalized Haar measure on $(\Qp^{N},+)$.
For $\rho=\infty$, $L^{\infty}(\Qp^{N})$ denotes the space of essentially
bounded measurable functions, with the norm
$\norm{g}_{\infty}=\operatorname*{ess\,sup}_{x\in\Qp^{N}}\abs{g(x)}$.
If $U$ is an open subset of $\Qp^{N}$, $\mathcal{D}(U)$ denotes the
$\C$-vector space of test functions with supports contained in $U$; then
$\mathcal{D}(U)$ is dense in
$L^{\rho}(U)=\{\varphi:U\rightarrow\C;\norm{\varphi}_{\rho}=\{\int_{U}\abs{\varphi(x)}^{\rho}d^{N}x\}^{1/\rho}<\infty\}$,
for $1\leq\rho<\infty$, see, e.g., \cite[Section 4.3]{Alberioetal}. We denote
by $L_{\R}^{\rho}(U)$ the real counterpart of $L^{\rho}(U)$.

\subsection{The Fourier transform}

We recall that a $p$-adic number $x\neq0$ has a unique expansion of the form
$x=p^{ord(x)}\sum_{j=0}^{\infty}x_{j}p^{j}$, where $x_{j}\in\{0,\dots,p-1\}$
and $x_{0}\neq0$. By using this expansion, we define the fractional part of
$x\in\Qp$, denoted $\{x\}_{p}$, as the rational number
\[
\{x\}_{p}=\begin{cases}0&\text{if }x=0\text{ or }ord(x)\geq0,\\
p^{ord(x)}\sum_{j=0}^{-ord(x)-1}x_{j}p^{j}&\text{if }ord(x)<0.\end{cases}
\]
Set $\chi_{p}(y)=\exp(2\pi\ii\{y\}_{p})$ for $y\in\Qp$. The map $\chi_{p}(\cdot)$
is an additive character on $\Qp$, i.e., a continuous map from $(\Qp,+)$ into
$S$ (the unit circle considered as a multiplicative group) satisfying
$\chi_{p}(x_{0}+x_{1})=\chi_{p}(x_{0})\chi_{p}(x_{1})$, $x_{0},x_{1}\in\Qp$.
It satisfies $\abs{\chi_{p}(y)}=1$, and $\chi_{p}(y)=1$ for $y\in\Zp$;
in particular, $\chi_{p}$ is bounded, but it does not belong to $L^{2}(\Qp)$.
The additive characters of $\Qp$ form an Abelian group which is isomorphic to
$(\Qp,+)$. The isomorphism is given by $\xi\mapsto\chi_{p}(\xi x)$, see,
e.g., \cite[Section 2.3]{Alberioetal}.

Set $\xi\cdot x:=\sum_{j=1}^{N}\xi_{j}x_{j}$, for $\xi=(\xi_{1},\dots,\xi_{N})$,
$x=(x_{1},\dots,x_{N})\in\Qp^{N}$, as before. The Fourier transform of
$\varphi\in\mathcal{D}(\Qp^{N})$ is defined as
\[
\F\varphi(\xi)=\int_{\Qp^{N}}\chi_{p}(\xi\cdot x)\varphi(x)d^{N}x\quad\text{for }\xi\in\Qp^{N},
\]
where $d^{N}x$ is the normalized Haar measure on $\Qp^{N}$. The Fourier
transform is a linear isomorphism from $\mathcal{D}(\Qp^{N})$ onto itself
satisfying
\begin{equation}
(\F(\F\varphi))(\xi)=\varphi(-\xi), \label{Eq_FFT}
\end{equation}
see, e.g., \cite[Section 4.8]{Alberioetal}. We also use the notation
$\F_{x\rightarrow\xi}\varphi$ and $\widehat{\varphi}$ for the Fourier
transform of $\varphi$. For instance,
$\F(\Om(p^{-r}\norm{x-a}_{p}))(\xi)=p^{rN}\chi_{p}(a\cdot\xi)\Om(p^{r}\norm{\xi}_{p})$.

The Fourier transform extends to $L^{2}(\Qp^{N})$. If $f\in L^{2}(\Qp^{N})$,
its Fourier transform is defined as
\[
(\F f)(\xi)=\lim_{k\rightarrow\infty}\int_{\norm{x}_{p}\leq p^{k}}\chi_{p}(\xi\cdot x)f(x)d^{N}x,\quad\text{for }\xi\in\Qp^{N},
\]
where the limit is taken in $L^{2}(\Qp^{N})$. We recall that the Fourier
transform is unitary on $L^{2}(\Qp^{N})$, i.e., $\norm{f}_{2}=\norm{\F f}_{2}$
for $f\in L^{2}(\Qp^{N})$, and that (\ref{Eq_FFT}) is also valid in
$L^{2}(\Qp^{N})$, see, e.g., \cite[Chapter III, Section 2]{Taibleson}.
Moreover, for $J\in L^{1}(\Qp^{N})$ and $\phi\in L^{2}(\Qp^{N})$,
$\F(J\ast\phi)=\widehat{J}\,\widehat{\phi}$.

\subsection*{Conflict of Interests/Competing Interests}

I have no conflicts of interest to disclose.


\begin{thebibliography}{99}

\bibitem{Beltrameti-et-al}Beltrametti E. G., Cassinelli G., Quantum mechanics
and $p$-adic numbers, Found. Phys. 2 (1972), 1--7.

\bibitem{V-V-QM1}Vladimirov V. S., Volovich I. V., $p$-adic quantum mechanics,
Soviet Phys. Dokl. 33 (1988), no. 9, 669--670.

\bibitem{V-V-QM2}Vladimirov V. S., Volovich I. V., A vacuum state in $p$-adic
quantum mechanics, Phys. Lett. B 217 (1989), no. 4, 411--415.

\bibitem{V-V-QM3}Vladimirov V. S., Volovich I. V., $p$-adic quantum
mechanics, Comm. Math. Phys. 123 (1989), no. 4, 659--676.

\bibitem{Volovich}Volovich I. V., Number theory as the ultimate physical
theory, $p$-Adic Numbers Ultrametric Anal. Appl. 2 (2010), no. 1, 77--87.

\bibitem{V-V-Z}Vladimirov V. S., Volovich I. V., Zelenov E. I., $p$-Adic
analysis and mathematical physics, World Scientific, 1994.

\bibitem{Dirac}Dirac P. A. M., The principles of quantum mechanics. First
edition 1930. Fourth Edition, Oxford University Press, 1958.

\bibitem{von-Neumann}von Neumann J., Mathematical foundations of quantum
mechanics. First Edition 1932. Princeton University Press, Princeton, NJ, 2018.

\bibitem{Zuniga-double-slit}Z\'{u}\~{n}iga-Galindo W. A., The $p$-adic
Schr\"{o}dinger equation and the two-slit experiment in quantum mechanics,
Annals of Physics 469 (2024), Paper No. 169747.

\bibitem{Zuniga-Dirac-Causality}Z\'{u}\~{n}iga-Galindo W. A., $p$-adic
quantum mechanics, the Dirac equation, and the violation of Einstein
causality, J. Phys. A 57 (2024), no. 30, Paper No. 305301, 29 pp.

\bibitem{Zuniga-2adic}Z\'{u}\~{n}iga-Galindo W. A., 2-Adic quantum mechanics,
continuous-time quantum walks, and the space discreteness, Fortschr. Phys.
(2025), e70019. https://doi.org/10.1002/prop.70019.

\bibitem{Zuniga-Mayes}Z\'{u}\~{n}iga-Galindo W. A., Mayes N. P., $p$-Adic
quantum mechanics, infinite potential wells, and continuous-time quantum
walks, Rev. Math. Phys. https://doi.org/10.1142/S0129055X25500278.

\bibitem{Zuniga-Chacon}Z\'{u}\~{n}iga-Galindo W. A., Chac\'{o}n-Cort\'{e}s
L. F., Continuous-time quantum Markov chains and discretizations of $p$-adic
Schr\"{o}dinger equations: comparisons and simulations, J. Stat. Mech. Theory
Exp. 2026 (2026), no. 9, Paper No. 093401, 55 pp.

\bibitem{Zuniga-JR}Z\'{u}\~{n}iga-Galindo W. A., $p$-Adic Dirac equations and
the Jackiw--Rebbi model, arXiv:2603.17200.

\bibitem{Zuniga-AP}Z\'{u}\~{n}iga-Galindo W. A., Quantum mechanics,
non-locality, and the space discreteness hypothesis, Annals of Physics 489
(2026), Paper No. 170459.

\bibitem{Zuniga-Wigner}Z\'{u}\~{n}iga-Galindo W. A., Wavefunctions
localization, and the Wigner's Friend Paradox in a Framework of Discrete-Space
Hypothesis, arXiv:2607.00198.

\bibitem{Zuniga-QNN}Z\'{u}\~{n}iga-Galindo W. A., Zambrano-Luna B. A.,
Indoung C., Pattern formation in quantum hierarchical cellular neural
networks, Commun. Nonlinear Sci. Numer. Simul. 163 (2026), part 5, Paper No.
110715.

\bibitem{Bialynicki-Birula}Bialynicki-Birula I., Weyl, Dirac, and Maxwell
equations on a lattice as unitary cellular automata, Phys. Rev. D 49 (1994),
no. 12, 6920--6927.

\bibitem{Meyer}Meyer D. A., From quantum cellular automata to quantum lattice
gases, J. Statist. Phys. 85 (1996), no. 5-6, 551--574.

\bibitem{Succi-Benzi}Succi S., Benzi R., Lattice Boltzmann equation for
quantum mechanics, Phys. D 69 (1993), no. 3-4, 327--332.

\bibitem{Arrighi-Nesme-Forets}Arrighi P., Nesme V., Forets M., The Dirac
equation as a quantum walk: higher dimensions, observational convergence, J.
Phys. A 47 (2014), no. 46, 465302, 33 pp.

\bibitem{Childs-Goldstone}Childs A. M., Goldstone J., Spatial search and
the Dirac equation, Phys. Rev. A 70 (2004), no. 4, 042312.

\bibitem{Bianconi}Bianconi G., The topological Dirac equation of networks
and simplicial complexes, J. Phys. Complex. 2 (2021), no. 3, 035022.

\bibitem{Todtli}T\"{o}dtli B., Laner M., Semenov J., Paoli B., Blattner M.,
Kunegis J., Continuous-time quantum walks on directed bipartite graphs,
arXiv:1606.00992.

\bibitem{Wallace}Wallace P. R., The band theory of graphite, Phys. Rev.
71 (1947), 622--634.

\bibitem{Castro-Neto}Castro Neto A. H., Guinea F., Peres N. M. R.,
Novoselov K. S., Geim A. K., The electronic properties of graphene, Rev.
Mod. Phys. 81 (2009), 109--162.

\bibitem{Jay-Debbasch-Wang}Jay G., Debbasch F., Wang J. B., Dirac quantum
walks on triangular and honeycomb lattices, Phys. Rev. A 99 (2019), no. 3,
032113, 11 pp.

\bibitem{Arnault-Debbasch}Arnault P., Debbasch F., Quantum walks and
discrete gauge theories, Phys. Rev. A 93 (2016), 052301.

\bibitem{Thaller}Thaller B., The Dirac equation, Texts Monogr. Phys.,
Springer-Verlag, Berlin, 1992.

\bibitem{Bjorken}Bjorken J. D., Drell S. D., Relativistic quantum
mechanics, McGraw-Hill Book Co., New York-Toronto-London, 1964.

\bibitem{Greiner}Greiner W., Relativistic quantum mechanics,
Springer-Verlag, Berlin, 2000.

\bibitem{Feynman-Hibbs}Feynman R. P., Hibbs A. R., Quantum Mechanics and Path
Integrals, McGraw-Hill, New York, 1965.

\bibitem{Strauch1}Strauch F. W., Relativistic quantum walks, Phys. Rev. A 73
(2006), 054302; erratum Phys. Rev. A 73 (2006), 069908.

\bibitem{Strauch2}Strauch F. W., Relativistic effects and rigorous limits for
discrete- and continuous-time quantum walks, J. Math. Phys. 48 (2007), no. 8,
082102, 18 pp.

\bibitem{Kurzynski}Kurzy\'{n}ski P., Relativistic effects in quantum walks:
Klein's paradox and Zitterbewegung, arXiv:quant-ph/0606171.

\bibitem{Chandrashekar}Chandrashekar C. M., Banerjee S., Srikanth R.,
Relationship between quantum walks and relativistic quantum mechanics, Phys.
Rev. A 81 (2010), no. 6, 062340, 10 pp.

\bibitem{Gerritsma1}Gerritsma R., Kirchmair G., Z\"{a}hringer F., Solano E.,
Blatt R., Roos C. F., Quantum simulation of the Dirac equation, Nature 463
(2010), 68--71.

\bibitem{Gerritsma2}Gerritsma R., Lanyon B. P., Kirchmair G., Z\"{a}hringer
F., Hempel C., Casanova J., Garc\'{\i}a-Ripoll J. J., Solano E., Blatt R.,
Roos C. F., Quantum simulation of the Klein paradox with trapped ions, Phys.
Rev. Lett. 106 (2011), no. 6, 060503, 4 pp.

\bibitem{Farhi-Gutman}Farhi E., Gutmann S., Quantum computation and decision
trees, Phys. Rev. A (3) 58 (1998), no. 2, 915--928.

\bibitem{Mulkne-Blumen}M\"{u}lken O., Blumen A., Continuous-time quantum
walks: models for coherent transport on complex networks, Phys. Rep. 502
(2011), no. 2-3, 37--87.

\bibitem{Venegas-Andraca}Venegas-Andraca S. E., Quantum walks: a
comprehensive review, Quantum Inf. Process. 11 (2012), no. 5, 1015--1106.

\bibitem{Childs-et-al}Childs A. M., Farhi E., Gutmann S., An example of the
difference between quantum and classical random walks, Quantum Information
Processing 1 (2002), 35--43. https://doi.org/10.1023/A:1019609420309

\bibitem{DiMolfetta-Debbasch}Di Molfetta G., Debbasch F., Discrete-time
quantum walks: continuous limit and symmetries, J. Math. Phys. 53 (2012), no.
12, 123302, 20 pp.

\bibitem{DiMolfetta-Brachet-Debbasch}Di Molfetta G., Brachet M., Debbasch F.,
Quantum walks as massless Dirac fermions in curved space-time, Phys. Rev. A 88
(2013), no. 4, 042301, 8 pp.

\bibitem{Alberioetal}Albeverio S., Khrennikov A. Yu., Shelkovich V. M.,
Theory of $p$-adic distributions: linear and nonlinear models, London
Mathematical Society Lecture Note Series, 370, Cambridge University Press, 2010.

\bibitem{Reed-Simon-I}Reed M., Simon B., Methods of modern mathematical
physics. I. Functional analysis, Second edition, Academic Press, New York, 1980.

\bibitem{Pazy}Pazy A., Semigroups of linear operators and applications to
partial differential equations, Applied Mathematical Sciences 44,
Springer-Verlag, New York, 1983.

\bibitem{Appell-Zabrejko}Appell J., Zabrejko P. P., Nonlinear
superposition operators, Cambridge Tracts in Mathematics 95, Cambridge
University Press, Cambridge, 1990.

\bibitem{GKS}Gorini V., Kossakowski A., Sudarshan E. C. G., Completely
positive dynamical semigroups of $N$-level systems, J. Math. Phys. 17 (1976),
821--825.

\bibitem{Lindblad}Lindblad G., On the generators of quantum dynamical
semigroups, Comm. Math. Phys. 48 (1976), 119--130.

\bibitem{Breuer-Petruccione}Breuer H.-P., Petruccione F., The theory of
open quantum systems, Oxford University Press, Oxford, 2002.

\bibitem{Zambrano-Zuniga-1}Zambrano-Luna B. A., Z\'{u}\~{n}iga-Galindo W. A.,
$p$-adic cellular neural networks, J. Nonlinear Math. Phys. 30 (2023), no. 1,
34--70.

\bibitem{Zambrano-Zuniga-2}Zambrano-Luna B. A., Z\'{u}\~{n}iga-Galindo W. A.,
$p$-adic cellular neural networks: applications to image processing, Phys. D
446 (2023), Paper No. 133668, 11 pp.

\bibitem{Zuniga-et-al}Z\'{u}\~{n}iga-Galindo W. A., Zambrano-Luna B. A.,
Dibba B., Hierarchical neural networks, $p$-adic PDEs, and applications to
image processing, J. Nonlinear Math. Phys. 31 (2024), 63.

\bibitem{Zuniga-Textbook}Z\'{u}\~{n}iga-Galindo W. A., $p$-Adic analysis:
stochastic processes and pseudo-differential equations, De Gruyter, 2025.

\bibitem{Taibleson}Taibleson M. H., Fourier analysis on local fields,
Princeton University Press, 1975.

\bibitem{Halmos}Halmos P., Measure theory, D. Van Nostrand Company Inc.,
New York, 1950.
\end{thebibliography}
\end{document}